\documentclass[xcolor=dvipsnames,12pt,reqno,a4paper]{amsart}
\usepackage{amsmath,amsfonts,amsthm,amssymb}
\newcommand\version{December 9, 2014}
\usepackage{amsxtra,mathrsfs,color}

\newtheorem{theorem}{Theorem}[section]
\newtheorem{proposition}[theorem]{Proposition}
\newtheorem{lemma}[theorem]{Lemma}

\theoremstyle{definition}

\theoremstyle{remark}

\numberwithin{equation}{section}

\newcommand{\C}{\mathbb{C}}

\renewcommand{\epsilon}{\varepsilon}

\newcommand{\F}{\mathcal{F}}

\newcommand{\N}{\mathbb{N}}

\renewcommand{\phi}{\varphi}
\newcommand{\R}{\mathbb{R}}

\newcommand{\Z}{\mathbb{Z}}

\DeclareMathOperator{\dist}{dist}

\DeclareMathOperator{\Tr}{Tr}
\DeclareMathOperator{\tr}{Tr}

\newcommand{\spinv}{{\vec S}}
\newcommand{\spin}{S}
\newcommand{\HH}{\mathscr{H}}
\newcommand{\OO}{\mathcal{O}}

\begin{document}

\title[Heisenberg Ferromagnet --- \version]
{Validity of the spin-wave approximation for the free energy of the Heisenberg ferromagnet}

\author{Michele Correggi} 
\address{Dipartimento di Matematica, ``Sapienza'' Universit\`{a} di Roma,	P.le Aldo Moro 5, 00185, Rome, Italy}
\email{michele.correggi@gmail.com}

\author{Alessandro Giuliani}
\address{Dipartimento di Matematica, Universit\`{a} degli Studi Roma Tre,	L.go S. Leonardo Murialdo 1, 00146, Rome, Italy}
\email{giuliani@mat.uniroma3.it}

\author{Robert Seiringer}
\address{Institute of Science and Technology Austria, Am Campus 1, 3400 Klosterneuburg, Austria}
\email{robert.seiringer@ist.ac.at}

\thanks{\copyright\, 2014 by the authors. This paper may be  
reproduced, in
its entirety, for non-commercial purposes.}

\begin{abstract} 
We consider the quantum ferromagnetic Heisenberg model in three dimensions, for all spins $S\ge 1/2$. We rigorously prove the validity of the spin-wave approximation for the excitation spectrum, at the level of the first non-trivial contribution to the free energy at low temperatures. Our proof comes with explicit, constructive upper and lower bounds on the error term. It uses in an essential way the bosonic formulation of the model in terms of the Holstein--Primakoff  representation. In this language, the model describes interacting bosons with a hard-core on-site repulsion and a nearest-neighbor attraction. This attractive interaction makes the lower bound on the free energy particularly tricky: the key idea there is to prove a differential inequality for the two-particle density, which is thereby shown to be smaller than the probability density of a suitably weighted two-particle random process on the lattice.
\end{abstract}

\date{\version}

\maketitle


\section{Introduction}

The spontaneous breaking of a continuous symmetry in statistical mechanics and field theory, even if well understood from a physical point of view, 
is still elusive in many respects as far as a rigorous mathematical treatment is concerned. The case of an abelian continuous symmetry is the easiest to handle, 
and for that a number of rigorous results are available, based on reflection positivity \cite{FSS, DLS}, possibly combined with 
a spin-wave expansion \cite{BFLLS}, or cluster expansion combined with a vortex loop representation \cite{FS, KK}. The non-abelian case is more subtle,
and the few results available are mostly based on reflection positivity\footnote{An exception is the work by Balaban on spontaneous symmetry breaking in classical $N$-vector models \cite{Ba}, which is based on rigorous renormalization group methods.}: see \cite{FSS} for the classical Heisenberg 
and \cite{DLS} for the quantum Heisenberg anti-ferromagnet. 

The ``standard" quantum model for the phenomenon of interest is the three-dim\-en\-sional quantum Heisenberg ferromagnet (QHFM), which is not reflection positive and eluded any rigorous treatment so far. At a heuristic level, its low-temperature thermodynamics, including 
a (formal) low temperature expansion for the free energy and the spontaneous magnetization, can be deduced from spin-wave theory \cite{B1,B2,D,D2,HK,HP},
but to date any attempt to put it on solid grounds failed. The only partial results available on the subject are, to the best of our knowledge: 
the upper bounds on the free energy of the $S=1/2$ QHFM by Conlon--Solovej \cite{CS2}  and by Toth \cite{T}, which are of the correct order at low temperatures, 
but off by a constant; the asymptotically correct upper and lower bounds on the free energy for large $S$ by two of us \cite{CG} (see also \cite{CS1,CS3,CS4} for earlier work). At large $S$, the effective attractive interaction in the bosonic picture (reviewed in Section~\ref{sec:bosons} below) is weak, of order $1/S$, simplifying the problem. The problem for finite $S$ is significantly harder; quite surprisingly, not even a sharp upper bound at low temperature was known so far. 

In this paper we give the first proof of asymptotic correctness of spin-wave theory for the QHFM for any fixed $S\geq 1/2$ in three dimension at zero external field, in the sense 
that we prove upper and lower bounds on the free energy that are asymptotically matching as $\beta\to\infty$, with explicit estimates on the 
error  {(see \cite{CGS} for a sketch of the proof in the case $S=1/2$)}. The method of proof uses an exact mapping of the model into a system of interacting bosons, via the well known Holstein-Primakoff representation \cite{HP}. Under 
this mapping, the Heisenberg model takes the form of an interacting system of bosons, the interaction including a hard-core term, which prevents more than $2S$
bosons to occupy a single site, as well as an attractive nearest neighbor contribution. 
Low temperatures correspond to low density in the boson language; therefore, 
the attractive interaction, even if not small, is expected to give a subleading contribution to the free energy at low temperatures, as compared to the kinetic energy 
term. A subtlety to keep in mind, which plays a role in the following proof, is that the bosonic representation apparently breaks the rotational invariance of the 
model. More precisely, the degenerate states in the quantum spin language are not obviously so in the bosonic one (rotational invariance is a hidden, rather than 
apparent, global symmetry of the model in the bosonic language). 

Our problem is reminiscent of the asymptotic computation of the ground state energy \cite{LY,LSY} and free energy \cite{SBose,JYin} of the low density Bose gas, but  new ideas are needed in order to deal with the attractive nature of the interaction, as well as with the non-abelian continuous symmetry of the problem. 

The rest of the paper is organized as follows: we first define the model and state the main results more precisely (Section~\ref{sec:model}). The representation of the Heisenberg model in terms of interacting bosons will be reviewed in Section~\ref{sec:bosons}, where we also present a key result concerning the two-point function of low-energy eigenfunctions of the Heisenberg Hamiltonian in Theorem~\ref{thm:rho}. The proofs of the upper bound (Section~\ref{sec:up}) and the lower bound (Section~\ref{sec:low}) to the free energy are given subsequently. Finally, Section~\ref{sec:rho} contains the proof of Theorem~\ref{thm:rho}. The proofs of auxiliary lemmas needed there are collected in an appendix.

Throughout the proofs, $C$ stands for unspecified universal constants. Constants with specific values will be denotes by $C_0$, $C_1$, \dots\ instead.


\section{Model and Main Result}\label{sec:model}

We consider the ferromagnetic Heisenberg model with nearest neighbor interactions on the cubic lattice $\Z^3$. It is defined in terms of the Hamiltonian
\begin{equation}\label{heisenberg ham 1}
H_\Lambda :=  \sum_{\langle x,y \rangle \subset \Lambda}(S^2- \spinv_{x} \cdot \spinv_y)\,,
\end{equation}
where $\Lambda$ is a finite subset of $\Z^3$, the sum is over all (unordered) nearest neighbor pairs $\langle x,y\rangle$ in $\Lambda$, and $ \spinv=(S^1,S^2,S^3)$ denote the three components of the spin operators corresponding to spin $S$, i.e., they are the generators of the rotations in a $2S+1$ dimensional representation of $SU(2)$. The Hamiltonian $H_\Lambda$ acts on the Hilbert space $\HH_\Lambda = \bigotimes_{x\in\Lambda} \C^{2S+1}$. We added a constant $S^2$ for every site in order to normalize the ground state energy of $H_\Lambda$ to zero.

Our main object of interest is the free energy per site
\begin{equation}\label{free energy}
f(S,\beta,\Lambda) : = - \frac{1}{\beta|\Lambda|} \ln  \tr_{\HH_{\Lambda}} \,  \exp \left( - \beta H_\Lambda \right)\, ,
\end{equation}
where $\beta$ denotes the inverse temperature, 
and its value in the thermodynamic limit 
\begin{equation}
f(S,\beta) : = \lim_{\Lambda \to \Z^3} f(S,\beta,\Lambda)\,.
\end{equation}
The limit has to be understood via a suitable sequence of increasing domains, e.g., cubes of side length $L$ with $L\to\infty$. We are interested in the behavior of  $f(S,\beta)$ in the low temperature limit $ \beta \to \infty $ for fixed $S$. A related question was addressed in \cite{CG}, where the large spin regime $ S \to \infty $ with $ \beta \propto S^{-1}$ was investigated.

We shall show that the free energy at low temperature can be well approximated by non-interacting spin-waves or magnons, i.e., free bosons. Our main result is as follows.

\begin{theorem} \label{thm:f}
For any $S\geq 1/2$, 
\begin{equation}\label{def:c0}
\lim_{\beta \to \infty} f(S,\beta) \beta^{5/2} S^{3/2} =  C_0 := \frac 1{(2\pi)^{3}} \int_{\R^3} \ln \left( 1- e^{- p^2}\right) dp  = - \frac{\zeta(5/2)}{8 \pi^{3/2}}\,,
\end{equation}
where $\zeta$ denotes the Riemann zeta function.
\end{theorem}

The convergence in (\ref{def:c0}) {is uniform as $\beta S\to \infty$, provided $ \beta S \ge S^\alpha$, for some $\alpha>0$}. 
The proof of Theorem~\ref{thm:f} will be given in Sections~\ref{sec:up} and~\ref{sec:low}. It comes with explicit upper and lower bounds on $f(S,\beta)$ which agree to leading order as $\beta S\to \infty$. The proof can be easily generalized to  lattice dimensions larger than $3$, but we restrict our attention to the three-dimensional case for simplicity.

We note that the low-temperature asymptotics of the free energy of the Heisenberg ferromagnet for $S=1/2$ has been studied previously by  Conlon and Solovej in \cite[Theorem 1.1]{CS2}, where an upper bound on $ f(1/2,\beta) $  of the form $ (\frac 12)^{-3/2} C_1 \beta^{-5/2} (1+o(1))$ was derived by means of a random walk representation of the Heisenberg model. However their coefficient $ C_1 $ in front of $ \beta^{-5/2} $ was not the optimal one,  
\begin{equation}
	 C_1 = - \frac 12 \frac 1{(2\pi)^3} \int_{\R^3} e^{-p^2} dp = -\frac{1}{16 \pi^{3/2}}\,.
\end{equation}
Later this result was improved by Toth in \cite[Theorem 1]{T} where it was shown that $C_1$ can be replaced $ C_2 = C_0 \ln 2$ in the upper bound. Here we not only improve these results by showing the optimal constant in the upper bound is $C_0$ for general $S$, we also provide a corresponding lower bound. 

An interesting consequence of our bounds is an instance of quasi long-range order, in the sense that, if $\langle\cdot\rangle_\beta$ is a translation invariant infinite volume Gibbs state
for the system at inverse temperature $\beta$,
\begin{equation}\label{2.or}
\langle S^2-\vec S_x\cdot \vec S_y\rangle_{\beta}\le \tfrac{9}{8}\|x-y\|^2_1 e(S,\beta)\,,
\end{equation}
where $e(S,\beta)=\partial_\beta(\beta f(S,\beta))$ is the energy per site. By concavity of the free energy, our upper and lower bounds on $f(S,\beta)$ imply similar bounds on $e(S,\beta)$,
via 
\begin{equation} \lambda e(S,\beta)\le f(S,\beta)-(1-\lambda) f(S,(1-\lambda)\beta)\;,\qquad \lambda\in(-\infty,1)\;.
\end{equation}
If we use  \eqref{fe ub asympt} and \eqref{fe lb asympt} below, and optimize over $\lambda$, we get
\begin{equation} e(S,\beta)= -\frac32 C_0S^{-3/2}\beta^{-5/2}(1+\mathcal O((S\beta)^{-\kappa}))\;,\qquad \kappa<\tfrac1{80}\;.
\end{equation}
Therefore, \eqref{2.or} implies that spin order persists up to length scales of the order $\beta^{5/4}$, in the sense that $\langle \vec S_x\cdot\vec S_y\rangle_\beta$
is bounded away from zero as long as $|x-y|\le ({\rm const.})\beta^{5/4}$. Spin wave theory predicts equality in \eqref{2.or}, without the factor $\tfrac{9}8$ and with the 
$\ell_1$ distance replaced by the euclidean one, asymptotically  for $|x-y|\ll \sqrt\beta$. 
Of course, one expects infinite range order at low temperatures, but in absence of a proof Eq.~\eqref{2.or} is the best result to
date. We shall prove \eqref{2.or} in Appendix \ref{appB}.

We conclude this section with a brief outline of the proof of Theorem~\ref{thm:f}.
To obtain an upper bound, we utilize the Gibbs variational principle. The natural trial state to use is the one of non-inter\-acting bosons, projected to the subspace 
where each site has occupation number at most $1$; for convenience the trial state is localized into boxes of suitable (temperature-dependent) size. 
A localization procedure is also used in the lower bound, whose proof is more sophisticated and roughly proceeds as follows: we first derive a ``rough" lower bound, off by logarithmic factors  
from the correct one, by localizing into boxes of side length $\ell \ll \beta^{1/2}$ and by using a basic lower bound on the excitation spectrum, scaling like 
$\ell^{-2}(S_{\rm max}-S^T)$, where $S^T$ is the total-spin quantum number, and $S_{\rm max}$ its maximal allowed value. This lower bound on the excitation spectrum has some interest in itself, and complements the sharp formula for the gap proved in \cite{CLR} in the spin $1/2$ case. Its method of proof is the key ingredient to 
get \eqref{2.or}.
Next, we move to a larger scale ($\ell\sim \beta^{1/2+\epsilon}$ for some small $\epsilon >0$): The preliminary rough bound allows us to discard states with large energy; by using rotational invariance, we can also restrict ourselves to computing the trace of interest in the subspace of lowest $3$-component of the total spin. On the corresponding subspace we then utilize the representation in terms of interacting bosons, and we use the Gibbs-Peierls-Bogoliubov inequality to estimate $-\ln \tr e^{-\beta H}$ from below by the non-interacting expression, minus the average of the interaction term. A bound on the latter will be presented in Theorem~\ref{thm:rho} in the next section, whose proof requires two key ideas:
(1) we use the eigenvalue equation to derive a suitable differential inequality for the two-particle density $\rho_2$, of the form $-\Delta \rho_2\le ({\rm const.})E\rho_2$, with $E$ the energy, which is a small number, and $\Delta$ a (modified) Laplacian on $\Lambda$; in this way we reduce the many-body problem to a two-body one; (2)
we iterate the inequality, thus obtaining an upper bound on $\|\rho_2\|_\infty$ in terms of the long-time probability density of a 
modified random walk on $\mathbb Z^6$.

\section{Boson Representation}\label{sec:bosons}

It is well known that the Heisenberg Hamiltonian can  be rewritten in terms of bosonic creation and annihilation operators \cite{HP}. For any $ x \in \Lambda $ we set
\begin{equation}\label{ax upx}
	\spin_{x}^+ =  \sqrt{2S}\, a^\dagger_x \left[ 1 - \frac{a^\dagger_x a_x}{2S} \right]_+^{1/2} \ ,		\quad	\spin_{x}^{-} = : \sqrt{2S}\left[ 1 - \frac{a^\dagger_x a_x}{2S} \right]_+^{1/2} a_x\ ,	\quad	\spin_{x}^3 = : a^\dagger_x a_x - S\,,
\end{equation}
where $a^{\dagger}_{x}, a_x$ are bosonic creation and annihilation operators,  $S^\pm = S^1 \pm i S^2$, and $[\, \cdot\,]_+ = \max\{0, \, \cdot\, \}$ denotes the positive part.  The operators $a^\dagger$ and $a$ act on the space $\ell^2(\N_0)$ via $(a\, f)(n) = \sqrt{n+1} f(n+1)$ and $(a^\dagger f)(n) = \sqrt{n} f(n-1)$, and satisfy the canonical commutation relations $[a,a^\dagger] = 1$. One readily checks that (\ref{ax upx}) defines a representation of $SU(2)$ of spin $S$, and the operators $\spinv_x$ leave the space $\bigotimes_{x\in\Lambda} \ell^2 ( [0,2S]) \cong \HH_\Lambda = \bigotimes_{x\in\Lambda} \C^{2S+1}$, which can be naturally identified with a subspace of the Fock space $\F:=\bigotimes_{x\in\Lambda} \ell^2(\N_0)$, invariant. 

The Hamiltonian $H_\Lambda$ in (\ref{heisenberg ham 1})  can be expressed in terms of the bosonic  creation and annihilation operators as
\begin{align} \nonumber
H_\Lambda  =   S  \sum_{\langle x,y\rangle\subset \Lambda} \biggl(  & - a^\dagger_x \sqrt{ 1 - \frac {n_x}{2S} }  \sqrt{ 1-\frac{n_y}{2S}}a_y  - a^\dagger_y \sqrt{ 1-\frac{n_y}{2S}}   \sqrt{ 1 - \frac {n_x}{2S} }  a_x \\ & + n_x + n_y - \frac 1{S} n_x n_y  \biggl) \,, \label{hamb}
\end{align}
where we denote the number of particles at site $x$ by $n_x= a^\dagger_x a_x$. It describes a system of bosons hopping on the lattice $\Lambda$, with nearest neighbor {\em attractive} interactions and a hard-core condition preventing more than $2S$ particles to occupy the same site. Also the hopping amplitude depends on the number of particles on neighboring sites, via the square root factors in the first line in (\ref{hamb}). Note that the resulting interaction terms are not purely two-body (i.e., they involve interactions between two or more 
particles; in other words, they are not just quartic in the creation-annihilation operators, but it involve terms with $6$, $8$, etc., operators). 

In the bosonic representation (\ref{hamb}), the vacuum is a ground state of the Hamiltonian, and the excitations of the model can be described as bosonic particles in the same way as phonons in crystals. There exists a zero-energy ground state for any particle number less or equal to $2S |\Lambda|$, in fact. While this may not be immediately apparent from the representation (\ref{hamb}), it is a result of the $SU(2)$ symmetry of the model. The total spin is maximal in the ground state, which is therefore $(2S|\Lambda|+1)$-fold degenerate, corresponding to the different values of the $3$-component of the total spin. The latter, in turn, corresponds to the total particle number (minus $S|\Lambda|$) in the bosonic language.

One of the key ingredients of our proof of the lower bound on $f$ is the following theorem, which shows that the two-point function of a low-energy eigenfunction of $H_\Lambda$ is approximately constant. Since this result may be of independent interest, we display it already at this point.

\begin{theorem}\label{thm:rho}
There exists a constant $C>0$ such that, if $\Psi$ is an eigenfunction of the Heisenberg Hamiltonian on $\Lambda_
\ell := [0,\ell)^3 \cap \Z^3$ with energy $E>0$, and 
\begin{equation}
\rho(x_1,x_2) = \langle\Psi | a^\dagger_{x_1} a^\dagger_{x_2} a_{x_2} a_{x_1} | \Psi \rangle
\end{equation}
is its two-particle density, then 
\begin{equation}\label{thm:eq:rho}
\|\rho\|_\infty \leq C  S^{-3} E^{3} \|\rho\|_1\;.
\end{equation}

\end{theorem}

The proof of Theorem~\ref{thm:rho} will be given in Section~\ref{sec:rho}. It will allow us to conclude that all terms in (\ref{hamb}) higher than quadratic in the creation and annihilation operators can be neglected at low energy, and the same is true for the constraint $n_x \leq 2S$. One is thus left with free bosons at zero chemical potential, whose free energy equals $C_0S^{-3/2}  \beta^{-5/2}$ (compare, e.g., with (\ref{4.28}) below). 

The bound~(\ref{thm:eq:rho}) can also be interpreted as absence of bound states of the bosons for small enough energy. It is well-known that due to the attractive nature of the nearest neighbor interaction bound states do exist   at higher energy, see \cite{bs2,bs3,bs1,GS}.


\section{Proof of Theorem~\ref{thm:f}; Upper Bound}\label{sec:up}

In this section we will prove the following.

\begin{proposition}\label{ub: pro}
Let $ C_0 $ be the constant given in \eqref{def:c0}. As $ \beta S \to \infty $, we have 
\begin{equation}\label{fe ub asympt}
f(S,\beta) \leq C_0 S^{-3/2} \beta^{-5/2} \left(1 - \OO( (\beta S)^{-3/8}) \right)\,.
\end{equation}
\end{proposition}

By the Gibbs variational principle, 
\begin{equation}\label{varpr}
f(S,\beta,\Lambda) \leq  \frac{1}{|\Lambda|} \tr H_\Lambda \Gamma  + \frac{1}{\beta |\Lambda|} \tr  \Gamma \ln \Gamma
\end{equation}
for any positive $\Gamma$ with $\tr \Gamma =1$. We can use this to confine particle into boxes, with Dirichlet boundary conditions. To be precise, let 
\begin{equation}
H_{\Lambda}^D = H_\Lambda + \sum_{\underset{|x-y|=1}{x\in \Lambda, y\in \Lambda^c}} \left( S^2 + S S_x^3\right)
\end{equation}
be the Heisenberg Hamiltonian on $\Lambda\subset \Z^3$ with $S^3_x=-S$ boundary conditions.  Note that $H_\Lambda^D \geq H_\Lambda$. We take $\Lambda$ to be the cube $\Lambda_L := [0,L)^3\cap \Z^3$ with $L^3$ sites, and assume that $L = k(\ell+1)$ for some integers $k$ and $\ell$. By letting all the spins point maximally in the negative $3$-direction on the boundary of the smaller cubes of side length $\ell$, we obtain the upper bound
\begin{equation}
f(S,\beta,\Lambda_L) \leq \left( 1 + \ell^{-1} \right)^{-3} f^D(S,\beta,\Lambda_\ell) \ , \quad f^D(S,\beta,\Lambda) := - \frac 1{\beta |\Lambda|} \ln \Tr e^{-\beta H_\Lambda^D} \,.
\end{equation}
In particular, by letting $k\to \infty$ for fixed $\ell$, we have 
\begin{equation}
f(S,\beta) \leq \left( 1 + \ell^{-1} \right)^{-3} f^D(S,\beta,\Lambda_\ell) 
\end{equation}
in the thermodynamic limit.

To obtain an upper bound on $f^D$, we can use the variational principle (\ref{varpr}), with
\begin{equation}
\Gamma = \frac{ \mathcal{P} e^{-\beta  T} \mathcal{P} }{\Tr_\F \mathcal{P} e^{-\beta T }} \,.
\end{equation}
Here,  $\mathcal{P}$ projects onto $n_x\leq 1$ for every site $x\in\Lambda_\ell$, and $T$ is the Hamiltonian on Fock space $\F$ describing free bosons on $\Lambda_\ell$ with Dirichlet boundary conditions, 
\begin{align}\nonumber 
T & =   S \sum_{x,y \in \Lambda_\ell} \left( -\Delta^D\right)(x,y) a^\dagger_x a_y \\ & =  S  \sum_{\langle x,y\rangle\subset \Lambda_\ell} \left( - a^\dagger_x a_y  - a^\dagger_y  a_x  + n_x + n_y  \right)  + S \sum_{\underset{|x-y|=1}{x\in \Lambda_\ell, y\in \Lambda_\ell^c}} n_x\,,\label{hamd}
\end{align}
where $\Delta^D$ denotes the Dirichlet Laplacian on $\Lambda_\ell$. The eigenvalues of $-\Delta^D$ are given by
\begin{equation}\label{epsd}
\left\{ \epsilon(p) = \sum_{j=1}^3 2 (1-\cos(p^j)) \, : \, p \in \Lambda_\ell^{*D}:= \left( \frac \pi{\ell+1} \{ 1, 2,\dots, \ell\} \right)^3 \right\}
\end{equation}
with corresponding eigenfunctions $\phi_p(x) = [2/(\ell+1)]^{3/2}\prod_{j=1}^3 \sin((x^j+1) p^j)$.

\begin{lemma} \label{lem:thp}
On the Fock space $\F=\bigotimes_{x\in\Lambda} \ell^2(\N_0)$, 
\begin{equation}\label{THP}
\mathcal{P} H_\Lambda^D \mathcal{P} \leq T +  (2S-1) \sum_{\langle x,y\rangle\subset \Lambda} n_x n_y  \,.
\end{equation}
\end{lemma}

Note that for $S=1/2$ the second term on the right side vanishes.

\begin{proof}
We write $\mathcal{P} = \prod_{x\in \Lambda} p_x$, where $p_x$ projects onto the subspace of $\F$ with $n_x\leq 1$. We have 
\begin{equation}
p_x p_y  a^\dagger_x \sqrt{ 1 - \frac {n_x}{2S} }  \sqrt{ 1-\frac{n_y}{2S}}a_y p_x p_y = p_x p_y a^\dagger_x a_y p_x p_y = a^\dagger_x p_x(1-n_x) (1-n_y)p_y a_y \,.
\end{equation}
In particular,
\begin{align}\nonumber
&\mathcal{P} \left(   - a^\dagger_x \sqrt{ 1 - \frac {n_x}{2S} }  \sqrt{ 1-\frac{n_y}{2S}}a_y  - a^\dagger_y \sqrt{ 1-\frac{n_y}{2S}}   \sqrt{ 1 - \frac {n_x}{2S} }  a_x + n_x + n_y - \frac 1{S} n_x n_y  \right) \mathcal{P}  \\ & = (a^\dagger_x - a^\dagger_y) \mathcal{P} (1-n_x)(1-n_y) (a_x-a_y) + \mathcal{P} \left( 2 - \frac 1 S\right) n_x n_y \,.
\end{align}
If we bound $\mathcal{P}(1-n_x)(1-n_y) \leq 1$ in the first term, and $\mathcal{P}\leq 1$ in the second, we arrive at (\ref{THP}).
\end{proof}

As a next step, we will show that $\tr_\F \mathcal{P} e^{-\beta T}$ is close to $\tr_\F e^{-\beta T}$ for  $\beta S \gg \ell$.

\begin{lemma}\label{lem:ppex}
With $C_3:= 8 \pi^{-3} \zeta(3/2)^2$, we have 
 \begin{equation}\label{ppex}
 \frac{ \Tr_\F \mathcal{P} e^{-\beta T}}{\Tr_\F e^{-\beta T} } \geq 1- \frac{C_3 \ell^3}{(\beta S)^3} \,.
 \end{equation} 
\end{lemma}

\begin{proof}
As in the proof of Lemma~\ref{lem:thp}, we write $\mathcal{P} = \prod_{x\in \Lambda_\ell} p_x$. Then
\begin{equation}\label{omp}
1-\mathcal{P}  \leq \sum_x (1-p_x) \leq \frac 12\sum_x  n_x (n_x-1) = \frac 12 \sum_x a^\dagger_x a^\dagger_x a_x a_x \,.
\end{equation}
Wick's rule for Gaussian states therefore implies that 
\begin{equation}
\frac{ \Tr_\F \mathcal{P} e^{-\beta T}}{\Tr_\F e^{-\beta T} } \geq 1- \frac 12 \sum_{x\in \Lambda_\ell} \frac{ \Tr_\F a^\dagger_x a^\dagger_x a_x a_x e^{-\beta T}}{\Tr_\F e^{-\beta T} } = 1 -\sum_{x\in \Lambda_\ell} \left( \frac{ \Tr_\F n_x e^{-\beta T}}{\Tr_\F e^{-\beta T} }\right)^2 \,.
\end{equation} 
Moreover,
\begin{equation}
 \frac{ \Tr_\F n_x e^{-\beta T}}{\Tr_\F e^{-\beta T}} = \frac 1{ e^{-\beta S \Delta^D} -1 }(x,x) = \sum_{n\geq 1}  e^{n\beta S \Delta^D}(x,x) \,.
 \end{equation}
It is well known that the heat kernel of the Dirichlet Laplacian $\Delta^D$ is pointwise bounded from above by the one of the Laplacian $\Delta_{\Z^3}$ on all of $\Z^3$; this follows, e.g., from the Feynman-Kac formula. The latter equals (see, e.g., \cite{CY})
 \begin{equation}
 e^{t \Delta_{\Z^3}}(x,x) = e^{-6t} I_0(2t)^3
 \end{equation}
 on the diagonal, with $I_0$ a modified Bessel function (see \cite{GR} or Eq.~(\ref{def:is}) below for a definition). As explained in (\ref{i0e}) below, $I_0(t) \leq  2e^t/\sqrt{\pi t}$, and thus 
  \begin{equation}
 \sum_{n\geq 1}  e^{n\beta S \Delta^D}(x,x)  \leq  \frac 8{(2\pi)^{3/2} } \frac 1 {(\beta S)^{3/2}}  \zeta(3/2)\,. 
 \end{equation}
 In particular, we obtain the bound (\ref{ppex}). 
\end{proof}

By using Wick's rule in the same way as in the proof of Lemma~\ref{lem:ppex}, and following the same estimates, we have, for $x\neq y$, 
\begin{align}\nonumber
\frac { \tr_\F n_x n_y e^{-\beta T}}{\Tr_\F e^{-\beta T}} & =   \frac { \tr_\F n_x e^{-\beta T}}{\Tr_\F e^{-\beta T}} \frac { \tr_\F  n_y e^{-\beta T}}{\Tr e^{-\beta T}} + \left(  \frac { \tr_\F a^\dagger_x a_y e^{-\beta T}}{\Tr_\F e^{-\beta T}}\right)^2 \\ & \leq 2  \frac { \tr_\F n_x e^{-\beta T}}{\Tr_\F e^{-\beta T}} \frac { \tr_\F  n_y e^{-\beta T}}{\Tr_\F e^{-\beta T}} \leq \frac{2C_3}{(\beta S)^3}\,,
\end{align}
where we used the Cauchy-Schwarz Inequality in the second step. In combination with Lemma~\ref{lem:thp} and~\ref{lem:ppex}, we have thus shown that 
\begin{equation}
\tr H_{\Lambda_\ell}^D \Gamma \leq \frac { \tr_\F T e^{-\beta T}}{\Tr_\F \mathcal{P} e^{-\beta T}}  + 12 (2S -1)\left(1- \frac{C_3 \ell^3}{(\beta S)^3} \right)^{-1} \frac {C_3\ell^3}{(\beta S)^3}\,,
\end{equation}
where we bounded the number of nearest neighbor pairs in $\Lambda_\ell$ by $6 \ell^3$.

It remains to give a bound on the entropy of $\Gamma$.

\begin{lemma}\label{lem:ent}
For some constant $C>0$ and $\ell \geq (\beta S)^{1/2} $
\begin{equation}\label{entb}
\frac 1 \beta \Tr \Gamma \ln \Gamma \leq - \frac 1 \beta \ln \tr_\F \mathcal{P} e^{-\beta T} -  \frac{\Tr_\F T e^{-\beta T} }{\Tr_\F \mathcal{P} e^{-\beta T} } + \frac {C}{(\beta S)^{9/2}} \frac {\ell^6}\beta   \frac{\Tr_\F e^{-\beta T} }{\Tr_\F \mathcal{P} e^{-\beta T} }\,.  
\end{equation}
\end{lemma}

\begin{proof}
We have 
\begin{equation}
\Tr \Gamma \ln \Gamma = - \ln \tr_\F \mathcal{P} e^{-\beta T} + \frac 1{\Tr_\F \mathcal{P} e^{-\beta T} } \Tr_\F \mathcal{P} e^{-\beta T} \mathcal{P} \ln \mathcal{P} e^{-\beta T} \mathcal{P}\,.
\end{equation}
Using the operator monotonicity of the logarithm, as well as the fact that the spectrum of $\mathcal{P} e^{-\beta T} \mathcal{P}$ and $e^{-\beta T/2} \mathcal{P} e^{-\beta T/2}$ agree, we can bound
\begin{align}\nonumber
 \Tr_\F \mathcal{P} e^{-\beta T} \mathcal{P} \ln \mathcal{P} e^{-\beta T} \mathcal{P} & =  \Tr_\F e^{-\beta T/2} \mathcal{P} e^{-\beta T/2} \ln e^{-\beta T/2} \mathcal{P} e^{-\beta T/2} \\ & \leq  \Tr_\F  e^{-\beta T/2} \mathcal{P} e^{-\beta T/2} \ln e^{-\beta T} = - \beta  \Tr_\F \mathcal{P} e^{-\beta T} T \,.
\end{align}
Hence
\begin{equation}
\Tr \Gamma \ln \Gamma \leq - \ln \tr_\F \mathcal{P} e^{-\beta T} - \beta \frac{\Tr_\F T e^{-\beta T} }{\Tr_\F \mathcal{P} e^{-\beta T} } + \beta  \frac{\Tr_\F T (1-\mathcal{P}) e^{-\beta T} }{\Tr_\F \mathcal{P} e^{-\beta T} } \,.
\end{equation}
In the last term, we can bound  $1-\mathcal{P}$ as in (\ref{omp}), and evaluate the resulting expression using Wick's rule. With $\phi_p$ the eigenfunctions of the Dirichlet Laplacian, displayed below Eq.~(\ref{epsd}), we obtain
\begin{align}\nonumber
\frac{ \Tr_\F T n_x(n_x-1) e^{-\beta T} }{\Tr_\F e^{-\beta T} } & =  \left(\frac{ \Tr_\F n_x e^{-\beta T} }{\Tr_\F e^{-\beta T} } \right)^2  \sum_{p \in \Lambda_\ell^{*D}}  \frac {2S \epsilon(p) }{e^{\beta  S \epsilon(p)} -1}   \\  & \quad + \frac{ \Tr_\F  n_x e^{-\beta T} }{\Tr_\F e^{-\beta T} }    \sum_{p \in \Lambda_\ell^{*D}}  \frac { S\epsilon(p)  |\phi_p(x)|^2 }{ \left( \sinh \tfrac 12 \beta S \epsilon(p) \right)^2 } \,.
\end{align}
The expectation value of $n_x$ can be bounded independently of $x$ by $\sqrt{C_3} (\beta S)^{-3/2}$, as in the proof of Lemma~\ref{lem:ppex}. When summing over $x$, we can use the normalization $\sum_x |\varphi_p(x)|^2 =1$. The sums over $p$ can be bounded by the corresponding integrals, which leads to the bound (\ref{entb}). 
\end{proof}

In combination, Lemmas~\ref{lem:thp}, \ref{lem:ppex} and \ref{lem:ent} imply that 
\begin{align}\nonumber
f^D(S,\beta,\Lambda_\ell)  & \leq - \frac 1 {\beta \ell^3} \ln \tr_\F e^{-\beta T}  - \frac 1 {\beta\ell^3} \ln\left(1- \frac{C_3 \ell^3}{(\beta S)^3} \right) \\ & \quad  + C \left(1- \frac{C_3 \ell^3}{(\beta S)^3} \right)^{-1} \left(\frac {\ell^3}{\beta (\beta S)^{9/2}} + \frac { 2S -1}{(\beta S)^3}\right)\,.
\end{align}
The first term on the right side equals
\begin{equation}
 - \frac 1 {\beta \ell^3} \ln \tr_\F e^{-\beta T} = \frac 1 {\beta\ell^3} \sum_{p\in \Lambda_\ell^{*D} } \ln ( 1- e^{-\beta S \epsilon(p)} )\,.
\end{equation}
By viewing the sum as a Riemann sum for the corresponding integral, it is not difficult to see that 
\begin{equation}\label{up:rie} 
 \frac 1 {\beta\ell^3} \sum_{p\in \Lambda_\ell^{*D} } \ln ( 1- e^{-\beta S \epsilon(p)} ) \leq  \frac 1 {\beta (2\pi)^3} \int_{[-\pi,\pi]^3} \ln ( 1- e^{-\beta S \epsilon(p)} ) + \frac C {S \beta^2 \ell } 
\end{equation}
for some constant $C>0$ (compare, e.g., with \cite[Lemma~4]{RSft}). We can further use that  $\epsilon(p)\leq |p|^2$ and find that
\begin{align}\nonumber
\frac 1 {\beta (2\pi)^3} \int_{[-\pi,\pi]^3} \ln ( 1- e^{-\beta S \epsilon(p)} ) & \leq \frac 1 {\beta (2\pi)^3} \int_{\R^3} \ln ( 1- e^{-\beta S |p|^2}  )+ \frac C {\beta (\beta S)^\alpha }  \\
& =C_0 S^{-3/2} \beta^{-5/2} + \frac C {\beta (\beta S)^\alpha }  \label{4.28}
\end{align}
for some $C>0$, $C_0$ defined in (\ref{def:c0}), and $\alpha>0$ arbitrary. For $\beta S \gg \ell \gg (\beta S)^{1/2}$, all the error terms are small compared to the main term. The optimal choice of $\ell$ is $\ell \sim (\beta S)^{7/8}$, which leads to the desired upper bound stated in (\ref{fe ub asympt}).

\section{Proof of Theorem~\ref{thm:f}; Lower Bound}\label{sec:low}

In this section we will prove the following lower bound on the free energy of the Heisenberg ferromagnet.

\begin{proposition}\label{lb: pro}
Let $C_0$ be the constant given in \eqref{def:c0}. {Given $\alpha>0$}, if $ \beta S \to \infty $ {with $ \beta S \ge S^\alpha$}, we have 
\begin{equation}\label{fe lb asympt}
f(S,\beta) \geq C_0 S^{-3/2} \beta^{-5/2} \left(1 + \OO( (\beta S)^{-\kappa}) \right)
\end{equation}
for any $\kappa<1/40$.
\end{proposition}

Let again denote $\Lambda_L = [0,L)^3\cap \Z^3$  a cube with $L^3$ sites, and let $L = k \ell$ for some positive integers $k$ and $\ell$. We can decompose $\Lambda_L$ into $k^3$ disjoint cubes, all of which are translations of $\Lambda_\ell$. By simply dropping the terms in the Hamiltonian (\ref{heisenberg ham 1}) corresponding to pairs of nearest neighbor sites in different cubes, we obtain the lower bound
\begin{equation}\label{subadd}
f(S,\beta,\Lambda_L) \geq f(S,\beta,\Lambda_\ell)\,.
\end{equation}
By sending $k\to \infty$ at fixed $\ell$, we thus have 
\begin{equation}
f(S,\beta) \geq f(S,\beta,\Lambda_\ell)
\end{equation}
for the free energy in the thermodynamic limit.

The Hamiltonian (\ref{heisenberg ham 1}) commutes with the total spin operators $\sum_{x\in \Lambda} \spinv_x$, and hence we can label all eigenstates by the value of the corresponding square of the total spin, i.e., by the integer or half-integer eigenvalues of $S^T$, where 
\begin{equation}\label{def:st}
S^T (S^T+1) = \left | \sum_{x\in \Lambda} \spinv_x \right|^2\,.
\end{equation}
The following proposition shows that $S^T$ is close to its maximal value $S \ell^3$ at low energy.

\begin{proposition}\label{ham lb: pro}
There exists a positive constant $ C > 0 $ such that
\begin{equation}\label{ham lb}
H_{\Lambda_\ell} \ \geq C \frac S{\ell^{2}} \left( S \ell^3 - S^T \right)\,.
\end{equation}
\end{proposition}

Note that the lower bound (\ref{ham lb}) implies, in particular, that the gap in the spectrum of $H_{\Lambda_\ell}$ above zero is at least as big as $CS \ell^{-2}$. Except for the value of the constant, this bound is sharp, since one can easily obtain an upper bound of the form $2S(1-\cos(\pi/\ell))\approx \pi^2 S \ell^{-2}$. This latter expression is actually known to be the exact gap in the spin $1/2$ case \cite{CLR} (see also \cite{gap1,gap2,gap3} for related results). 

\begin{proof}
The starting point is the simple inequality
\begin{equation}\label{ineq 1}
\left( S^2  - \spinv_{x} \cdot \spinv_{y} \right) + \left( S^2 - \spinv_{y} \cdot\spinv_{z} \right) \geq \frac{1}{2} \left( S^2 - \spinv_{x} \cdot \spinv_{z} \right)
\end{equation}
for distinct sites $x$, $y$ and $z$. To prove it, it is convenient to use the equivalent formulation
\begin{equation}
S^2 - \frac 12 S   - \spinv_{y} \cdot \left(  \spinv_{x} + \spinv_{z} \right) + \frac{1}{4} \left( \spinv_{x} + \spinv_{z} \right)^2 \geq 0\,.
\end{equation}
The eigenvalues of $( \spinv_{x} + \spinv_{z} )^2$ are given by $t(t+1)$, with $t \in\{ 0, 1, \dots, 2S\}$, and we have  $ - \spinv_{y} \cdot ( \spinv_{x} + \spinv_{z} ) \geq  - S  t$ in the subspace corresponding to $t$. It is thus sufficient to prove that 
\begin{equation}
S^2 - \frac 12 S - S t + \frac 14  t( t+1) \geq 0  \quad \forall t \in \{ 0, 1, \dots, 2S\}\,.
\end{equation}
The expression on the left side of  this inequality vanishes for $t=2S$ and $t=2S-1$, and since it is quadratic in $t$ this implies non-negativity for all the relevant $t$. This proves (\ref{ineq 1}). 

We claim that if we have a number $n+1$ of distinct sites $x_j$,  inequality (\ref{ineq 1}) implies that 
\begin{equation}\label{indi}
\sum_{j=1}^{n}  \left( S^2  - \spinv_{x_{j}} \cdot \spinv_{x_{j+1}} \right)  \geq \frac{1}{2n}  
 \left( S^2 - \spinv_{x_{1}} \cdot \spinv_{x_{n+1}} \right)\,.
\end{equation}
If $n=2^k$ for some $k\geq 1$, this follows immediately from a repeated application of (\ref{ineq 1}), even without the factor $2$ in the denominator on the right side. The result in the general case can then easily be obtained by induction, writing a general $n$ as a sum $n = \sum_{j=1}^m 2^{k_j}$ with $0\leq k_1 < k_2 < \dots < k_m$.

For any pair of distinct sites $(x,y)\in \Lambda_\ell\times \Lambda_\ell$, we choose a path $x_0,x_1,\dots,x_n$ in $\Lambda_\ell$, such that $x_0=x$, $x_n=y$, $|x_{j-1} - x_j|=1$ for all $1\leq j \leq n$, and $x_j \neq x_{k}$ for $k\neq j$. Then (\ref{indi}) implies that 
\begin{equation}\label{indi2}
S^2 - \spinv_{x} \cdot \spinv_{y} \leq 2 n  \sum_{j=1}^n \left( S^2 - \spinv_{x_{j-1}} \cdot \spinv_{x_j} \right)\,.
\end{equation}
We shall choose the path as short as possible, i.e., $n=\|x-y\|_1\leq 3 \ell$. There are many such paths, and we take one that is closest to the straight line connecting $x$ and $y$. Let us denote such a path by $\mathcal{C}_{x,y}$. We have 
\begin{align}\nonumber 
S  \ell^3 \left( S \ell^3 + 1 \right) - S^T(S^T + 1) & = \sum_{\underset{x \neq y}{x, y \in \Lambda_\ell}} \left(  S^2 - \spinv_{x} \cdot \spinv_{y} \right) 	\\ \nonumber 
& \leq 2  \sum_{\underset{x \neq y}{x, y \in \Lambda_\ell}} \| x-y\|_1 \sum_{(x_i,x_{i+1})\in \mathcal{C}_{x,y}}  \left( S^2 - \spinv_{x_i} \spinv_{x_{i+1}} \right) \\ &
\leq 6 \ell  \sum_{\underset{|x-y|=1}{x, y \in \Lambda_\ell}}  \left(  S^2 - \spinv_{x} \cdot \spinv_{y} \right)  N_{x,y}\,, \label{nxy}
\end{align}
where $N_{x,y}$ denotes the number of  paths among all the $\mathcal{C}_{z,z'}$, $z,z'\in \Lambda_\ell$, that contain the step $x\to y$. By construction, the edge from $x$ to $y$ can be part of $\mathcal{C}_{z,z'}$ only if either $x$ or $y$ is within a distance $\mathcal{O}(1)$ from the line connecting $z$ and $z'$. For a given $z \neq x$, this will be the case for at most $C \ell^3 |x-z|^{-2}$ values of $z'$, which leads to the bound $N_{x,y} \leq C \ell^4$ for some $C>0$ for all nearest neighbor pairs $(x,y)$. By inserting this bound in (\ref{nxy}), we thus obtain 
\begin{align}\nonumber
S\ell^3 \left( S \ell^3 - S^T\right) & \leq S  \ell^3 \left( S \ell^3 + 1 \right) - S^T(S^T + 1) 
\\ & \leq 6 C \ell^5  \sum_{\underset{|x-y|=1}{x, y \in \Lambda_\ell}}  \left(  S^2 - \spinv_{x} \cdot \spinv_{y} \right)  = 12 C \ell^5 H_{\Lambda_\ell}\,.
\end{align}
This completes the proof of (\ref{ham lb}).
\end{proof}

With the aid of the bound (\ref{ham lb}) we shall now prove the following preliminary lower bound on the free energy.

\begin{lemma} \label{lem:prel}
For $\ell \geq (\beta S)^{1/2}$ and {$\beta S \ge S^\alpha$, with $\alpha>0$}, we have
\begin{equation}\label{prel lb 1}
f(S,\beta,\Lambda_\ell) \geq - C S  \left( \frac{\ln S\beta}{S\beta} \right)^{5/2}
\end{equation}
for some constant $ C{=C(\alpha)}>0$.
\end{lemma}

\begin{proof}
The dimension of the subspace of $\HH_{\Lambda_\ell}=\bigotimes_{x\in \Lambda_\ell} \C^{2S+1}$ corresponding to $S^T = \ell^3 S - s$ is bounded from above by 
\begin{equation}
(2 \ell^3 S  +1) \binom{2S \ell^3}{s}\,.
\end{equation}
The factor $2\ell^3S+1$ is a bound on the number of different values of the $3$-component of the total spin, and the binomial factor comes from distributing the $s$ particles over $2S \ell^3$ slots, $2S$ for each site. 
Hence, from (\ref{ham lb}), 
\begin{align}\nonumber
\tr e^{-\beta H_{\Lambda_\ell}} &\leq   \tr e^{\beta C S \ell^{-2} (S^T - \ell^3 S)} \leq (2 \ell^3 S +1)\sum_{s=0}^{\lfloor S \ell^3 \rfloor}  \binom{2S \ell^3}{s} e^{-\beta C S \ell^{-2} s} \\ & \leq 
(2 \ell^3 S +1) \left( 1+  e^{-\beta C S \ell^{-2}} \right)^{2S \ell^3}\,.
\end{align}
The free energy is thus bounded from below as
\begin{align}\nonumber
f(S,\beta,\Lambda_\ell) & \geq - \frac {2S} \beta  \ln \left( 1+  e^{-\beta C S \ell^{-2}} \right) - \frac 1 {\beta \ell^3} \ln \left(2 \ell^3 S +1\right) \\ &\geq - \frac{2S}\beta   e^{-\beta C S \ell^{-2}}  - \frac 1{\beta \ell^3} \ln \left(2 \ell^3 S +1\right) \,.
\end{align}
For \begin{equation}
\ell = \ell_0 := ( \beta C S)^{1/2} \left(\ln \left( S(\beta C S)^{3/2}\right) \right)^{-1/2}
\end{equation}
this yields an expression of the desired form (\ref{prel lb 1}). For larger $\ell$, we can use the subadditivity (\ref{subadd}) to obtain the result in general. 
\end{proof} 

We now come to the main part of our lower bound on the free energy. The preliminary estimate \eqref{prel lb 1} allows us to restrict the computation of the partition function to the subspace of states with not too large energy. Let $P_E$ be the spectral projection of $H_{\Lambda_\ell}$ corresponding to energy $\leq E$. Then
\begin{equation}
\tr (1-P_E) e^{-\beta H_{\Lambda_\ell}} \leq e^{-\beta E/2} \Tr (1-P_E) e^{-\beta H_{\Lambda_\ell}/2} \leq e^{-\beta E/2} e^{-\beta \ell^3 f(S,\beta/2,\Lambda_\ell) /2} \,.
\end{equation}
In particular, with 
\begin{equation}\label{def:e0}
E = E_0 := - {\ell^3}  f(S,\beta/2,\Lambda_\ell) \,,
\end{equation}
we have
\begin{equation}\label{p1}
\tr (1-P_{E_0}) e^{-\beta H_{\Lambda_\ell}} \leq 1 \,.
\end{equation}
Note that Lemma~\ref{lem:prel} implies that 
\begin{equation}\label{upe0}
E_0 \leq C \ell^3 S^{-3/2} (\beta^{-1} \ln S\beta)^{5/2} \qquad \text{for $\ell \geq ( \beta S)^{1/2}$.}
\end{equation} 

For the part of the spectrum corresponding to energy $\leq E_0$, we decompose the Hilbert space into sectors of total spin $S^T$, defined in (\ref{def:st}). For given $S^T$, every eigenvalue of $H_{\Lambda_\ell}$ is $(2S^T+1)$-fold degenerate, corresponding to the different values $-S^T, S^T+1, \dots, S^T$ the third component of the total spin,  $\sum_{x\in\Lambda_\ell} S_x^3$,  can take. We can thus restrict our attention to the eigenstates for which $\sum_{x\in\Lambda_\ell} S_x^3 = -S^T$, taking the degeneracy factor into account. That is, with $P^3$ denoting the projection onto the subspace of our Hilbert space corresponding to $\sum_{x\in\Lambda_\ell} S_x^3 = -S^T$, we have
\begin{equation}
\Tr g (H_{\Lambda_\ell}) = \Tr\, (2 S^T +1) P^3 g(H_{\Lambda_\ell})
\end{equation}
for any function $g$. In particular,
\begin{equation}\label{p2}
\tr P_{E_0} e^{-\beta H_{\Lambda_\ell}} = \Tr P_{E_0} (2 S^T +1) P^3 e^{-\beta H_{\Lambda_\ell}} \leq (2 S \ell^3 +1)\Tr P_{E_0}  P^3 e^{-\beta H_{\Lambda_\ell}} \,.
\end{equation}
Note the total particle number in any eigenstate of $H_{\Lambda_\ell}$ in the range of $P_{E_0} P^3$ is bounded by 
$\ell^2 E_0 /(CS)$, according to Proposition~\ref{ham lb: pro}.

Let us denote $P_{E_0} P^3$ by $Q_{E_0}$ for short. By combining (\ref{p1}) and (\ref{p2}), we obtain
\begin{equation}
\Tr e^{-\beta H_{\Lambda_\ell}} \leq 1 + (2S\ell^3+1 ) \Tr Q_{E_0} e^{-\beta H_{\Lambda_\ell}} \leq  (2S\ell^3+2 ) \Tr Q_{E_0} e^{-\beta H_{\Lambda_\ell}} \,,
\end{equation}
where we have used that $\Tr Q_{E_0} e^{-\beta H_{\Lambda_\ell}} \geq 1$ in the last step (which follows from the fact that $H_{\Lambda_\ell}$ has a zero eigenvalue with eigenvector in the range of $Q_{E_0}$). 
If we write $H_{\Lambda_\ell} = T - K$ for two hermitian operators $T$ and $K$, the Peierls--Bogoliubov inequality implies that
\begin{equation}
\tr Q_{E_0} e^{-\beta H_{\Lambda_\ell}} \leq \tr Q_{E_0} e^{-\beta Q_{E_0} T Q_{E_0} } \exp \left( \beta \frac { \tr Q_{E_0} K Q_{E_0} e^{-\beta H_{\Lambda_\ell}}}{ \tr Q_{E_0} e^{-\beta H_{\Lambda_\ell}}}\right)\,.
\end{equation}
We choose $T$ to be the Hamiltonian of free bosons, projected to our Hilbert space where $n_x \leq 2S$ for every $x \in \Lambda_\ell$. That is, 
\begin{equation}
T  =   S   \sum_{\langle x,y\rangle \subset \Lambda_\ell} \mathcal{P}_S  \left(  - a^\dagger_x a_y  - a^\dagger_y   a_x  + n_x + n_y \right) \mathcal{P}_S  \label{def:T}
\end{equation}
with $\mathcal{P}_S$ the projection onto $n_{x}\leq 2S$ for every site. The operator $K$ is then simply defined via $H_{\Lambda_\ell} = T - K$.  We have the following bound on $K$, similar to \cite[Prop.~2.3]{CG}. 

\begin{lemma}\label{lem:K}
The operator $K$ defined above satisfies the bound
\begin{equation}
K \leq \frac 12 \sum_{\langle x,y\rangle \subset \Lambda_\ell}  \left( 4 n_x n_y + n_x (n_x -1) +  n_y( n_y -1) \right)\,.
\end{equation}
\end{lemma}

\begin{proof}
The operator $K$ can be written as a sum of two terms, $K=K_1 + K_2$, with 
\begin{equation}
K_2 =  \sum_{\langle x,y\rangle \subset \Lambda_\ell} n_x n_y \,.
\end{equation}
Hence it only remains to look at $K_1$, given by 
\begin{equation}
K_1 =  - S   \sum_{\langle x,y\rangle \subset \Lambda_\ell}  \mathcal{P}_S \left( a^\dagger_x  k_{x,y} a_y + a^\dagger_y k_{x,y} a_x   \right) \mathcal{P}_S \,,
\end{equation}
where
\begin{equation}
k_{x,y} :=   1-  \sqrt{ 1 - \frac {n_x}{2S} }  \sqrt{ 1-\frac{n_y}{2S}} \geq 0 \,.
\end{equation}
The Cauchy-Schwarz inequality and the fact that $k_{x,y} \leq (n_x+n_y)/(2S)$ imply that 
\begin{align}\nonumber 
K_1 & \leq  S   \sum_{\langle x,y\rangle \subset \Lambda_\ell}  \mathcal{P}_S \left( a^\dagger_x  k_{x,y} a_x + a^\dagger_y k_{x,y} a_y   \right) \mathcal{P}_S 
\\ & \leq \frac 12   \sum_{\langle x,y\rangle \subset \Lambda_\ell}  \mathcal{P}_S \left( n_x(n_x-1) + n_y(n_y-1) + 2 n_x n_y   \right) \mathcal{P}_S \,.
\end{align}
The projections $\mathcal{P}_S$ can be dropped in  the last expression, since $\HH_{\Lambda_\ell} = \mathcal{P}_S \F$ is left invariant by the operators $n_x$.
\end{proof}

Let now $\Psi$ be an eigenstate of $H_{\Lambda_\ell}$ in the range of $Q_{E_0}$, and let $\rho(x_1,x_2) = \langle \Psi | a^\dagger_x a^\dagger_y a_y a_x |\Psi\rangle$ denote its two-particle density. From Lemma~\ref{lem:K} we have 
\begin{equation} 
\langle \Psi | K | \Psi\rangle \leq \sum_{\langle x,y\rangle \subset \Lambda_\ell}  \left( 2  \rho(x,y) + \frac 12 \rho(x,x) + \frac 12 \rho(y,y) \right) \leq 18 \ell^3 \|\rho\|_\infty \,.
\end{equation}
Theorem~\ref{thm:rho} states that $\|\rho\|_\infty \leq C S^{-3}E_0^{3} \|\rho\|_1$. 
Moreover,  $\|\rho\|_1$ is bounded by the square of the particle number, i.e., 
\begin{equation}
\|\rho\|_1 \leq  \left( \frac{\ell^2 E_0}{CS} \right)^2\,.
\end{equation}
 In particular, we conclude that 
\begin{equation}
 \frac { \tr Q_{E_0} K Q_{E_0} e^{-\beta H_{\Lambda_\ell}}}{ \tr Q_{E_0} e^{-\beta H_{\Lambda_\ell}}} \leq \frac {C}{S^5}  \ell^7 E_0^5 \leq   C  \frac {\ell^2}{(S\beta)^{5/2}} \left( \frac {\ell^2}{S\beta} \right)^{10}  ( \ln S\beta )^{25/2}
 \end{equation}
for some constant $C>0$ and $\ell \geq (\beta S)^{1/2}$.

We are left with deriving an upper bound on 
\begin{equation}
 \tr Q_{E_0} e^{-\beta Q_{E_0} T Q_{E_0} }\,,
\end{equation}
with $T$ defined in (\ref{def:T}) above. As already noted, the total number of particles in the range of $Q_{E_0}$ is bounded by $N_0:=\ell^2 E_0 /(CS)$, and hence $Q_{E_0} \leq \mathcal{Q}_{N_0}$, the projection onto the subspace corresponding to particle number $\leq N_0$. Let again $\F = \bigotimes_{x\in \Lambda_\ell} L^2(\N_0)$ denote the bosonic Fock space. The operator $T$ in (\ref{def:T}) is of the form $\mathcal{P}_S T_0 \mathcal{P}_{S}$, with $T_0$ the Hamiltonian for free bosons on $\F$. We can thus write
\begin{equation}
 \tr Q_{E_0} e^{-\beta Q_{E_0} T Q_{E_0} } = \Tr_\F \mathcal{P}_S Q_{E_0} e^{-\beta  Q_{E_0}\mathcal{P}_S T_0 \mathcal{P}_S Q_{E_0}}  
 \end{equation}
 where we denote by $\Tr_\F$ the trace on the Fock space $\F$. By the Gibbs variational principle, 
 \begin{align}\nonumber
& - \frac 1 \beta   \ln \Tr_\F \mathcal{P}_S Q_{E_0} e^{-\beta  Q_{E_0}\mathcal{P}_S T_0 \mathcal{P}_S Q_{E_0}}  \\ & = \min \left\{ \Tr T_0 \rho + \frac 1\beta \Tr_\F \rho \ln \rho \, : \, 0\leq \rho \leq \mathcal{P}_S Q_{E_0} \, , \ \Tr_\F \rho =1 \right\}\,.
 \end{align}
Since  $\mathcal{P}_{S} Q_{E_0} \leq \mathcal{Q}_{N_0}$ (viewed as an operator on $\F$), this implies that 
 \begin{equation}
\Tr_\F \mathcal{P}_S Q_{E_0} e^{-\beta  Q_{E_0}\mathcal{P}_S T_0 \mathcal{P}_S Q_{E_0}}   \leq    \Tr_\F \mathcal{Q}_{N_0} e^{-\beta \mathcal{Q}_{N_0} T_0 \mathcal{Q}_{N_0} } = \Tr_\F \mathcal{Q}_{N_0} e^{-\beta T_0 }\,,
\end{equation}
where we used that $\mathcal{Q}_{N_0}$ commutes with $T_0$ in the last step.

The eigenvalues of the Laplacian on $\Lambda_\ell$ are given by 
\begin{equation}
\left\{ \epsilon(p) = \sum_{i=1}^3 2 (1- \cos(p^i)) \, : \, p \in \Lambda_\ell^{*N} :=   \frac \pi \ell \Lambda_\ell \right\}\,.
\end{equation}
For $p\neq 0$, we can simply ignore the restriction on the particle number, and bound
\begin{equation}
\Tr_\F \mathcal{Q}_{N_0} e^{-\beta T_0 } \leq (N_0+1) \prod_{\underset{p\neq 0}{p \in \Lambda^{*N}_\ell}} \frac 1 {1- e^{-\beta S \epsilon(p)}}\,.
\end{equation}
By viewing the sum over $p$ is a Riemann approximation to the corresponding integral, it is not difficult to see that 
\begin{equation}
 \frac 1 {\beta\ell^3} \sum _{\underset{p\neq 0}{p \in \Lambda^{*N}_\ell}} \ln \left(1- e^{-\beta S \epsilon(p)}\right)  \geq  \frac 1 {(2\pi)^3\beta} \int_{[-\pi,\pi]^3}\ln \left(1- e^{-\beta S \epsilon(p)}\right) dp  - \frac C {S \beta^2 \ell } 
 \end{equation}
for some constant $C>0$. (Compare with the corresponding bound (\ref{up:rie}) in the previous section.) Finally, for some $C>0$ 
\begin{equation}\label{tos} 
 \frac 1 {(2\pi)^3\beta} \int_{[-\pi,\pi]^3}\ln \left(1- e^{-\beta S \epsilon(p)}\right) dp   \geq   \frac {C_0}{\beta^{5/2}S^{3/2}} \left( 1 + \frac C {\beta S} \right)
\end{equation}
with $C_0$ given in (\ref{def:c0}). To see (\ref{tos}), note that $C_0 \beta^{-5/2} S^{-3/2}$ is the result of the integral if $\epsilon(p)$ is replaced by $|p|^2$ and the region of integration $[-\pi,\pi]^3$ is replaced by $\R^3$. Using that $\epsilon(p)\geq |p|^2 \max\{ 1 - |p|^2/12, 4 /\pi^2\}$ for $p \in [-\pi,\pi]^3$, we have
\begin{align}\nonumber 
&  \frac 1 {(2\pi)^3\beta} \int_{[-\pi,\pi]^3}\ln \left(1- e^{-\beta S \epsilon(p)}\right) dp   - \frac {C_0}{\beta^{5/2}S^{3/2}} \\ & \geq     \frac 1 {(2\pi)^3\beta} \int_{|p|\leq 2}\ln \frac{1- e^{-\beta S |p|^2(1-|p|^2/12)}}{1- e^{-\beta S |p|^2}}  dp  +  \frac 1 {(2\pi)^3\beta} \int_{|p|\geq 2 }\ln \left(1- e^{- 4\beta S |p|^2/\pi^2}\right) dp \,.
\label{tos2} 
\end{align}
The last term is exponentially small in $\beta S$. In the integrand of the first term, we can bound 
\begin{equation}
\ln \frac{1- e^{-\beta S |p|^2(1-|p|^2/12)}}{1- e^{-\beta S |p|^2}} = - \int_0^{\beta S|p|^4/12} \frac 1{ e^{\beta S |p|^2 - t} -1} dt \geq -\frac {\beta S |p|^4}{12} \frac 1{ e^{2 \beta S |p|^2/3} -1}
\end{equation}
for $|p|\leq 2$, which leads to the desired estimate (\ref{tos}).

Collecting all the bounds, we have
\begin{align}\nonumber
f(S,\beta,\Lambda_\ell) & \geq \frac {C_0}{\beta^{5/2}S^{3/2}} \left( 1 + \frac C {\beta S} \right) -  \frac {C}{\ell} \frac {1}{(S\beta)^{5/2}} \left( \frac {\ell^2}{S\beta} \right)^{10}  ( \ln S\beta )^{25/2} \\ & \quad -  \frac C {S \beta^2 \ell }   - \frac 1 {\beta \ell^3} \ln \left[ (N_0+1) ( 2 S \ell^3 + 2)\right]\,.
\end{align} 
We are still free to choose $\ell$. For the choice $\ell = (\beta S)^{21/40}$ we obtain an error term smaller than $C S (\beta S)^{-5/2 - 1/40} (\ln \beta S)^{25/2}$, implying (\ref{fe lb asympt}).

\section{Proof of Theorem~\ref{thm:rho}}\label{sec:rho}

In this section we will give the proof of Theorem~\ref{thm:rho}. Note that since $\|\rho\|_\infty \leq \|\rho\|_1$ holds trivially, it suffices to prove the theorem when the parameter $E/S$ 
is suitably small. Thanks to Proposition~\ref{ham lb: pro}, all non-zero eigenvalues of $H_{\Lambda_\ell}$ are bounded from below by $C S \ell^{-2}$. Hence 
$E/S$ small implies that $\ell$ is large. 

We shall divide the proof into several steps. In Step 1, we shall prove a differential inequality satisfied by the two-particle density of an eigenstate of $H_{\Lambda}$. It involves a suitable weighted Laplacian on $\Lambda\times \Lambda$. In Step 2, we shall use the method of reflections to extend the inequality from $\Lambda_\ell\times \Lambda_\ell$ to the whole of $\Z^6$. By iterating the resulting inequality, we obtain a bound on the two-particle density in terms of the probability density of a random walk on $\Z^6$. The necessary bounds on this probability density are derived in Step 3. With their aid, we can show that the desired bound on the two-particle density holds at least a certain finite distance away from the boundary of $\Lambda_\ell\times\Lambda_\ell$. To extend this result to the whole space, we shall show in Step 4 that our differential inequality also implies that the two-particle density is very flat near its maximum, implying that 
its maximal value in the smaller cube a finite distance away from the boundary is very close to its global maximum.

\subsection{Step 1}
The first step in the proof is to derive a differential inequality for the two-particle density of an eigenstate of $H_{\Lambda}$. We state it in the following lemma.

\begin{lemma} Let $\Psi$ be an eigenstate of $H_\Lambda$ with eigenvalue $E$, and let $\rho(x_1,x_2) = \langle \Psi | a^\dagger_{x_1} a^\dagger_{x_2} a_{x_2} a_{x_1}|\Psi\rangle$ denote its two-particle density. Then
\begin{align}\nonumber
\frac {2 E}{S} \rho(x_1,x_2)   & \geq \sum_{\underset{|y-x_1|=1}{y\in \Lambda}} \left[ \rho(x_1, x_2) \left( 1 - \frac {\delta_{y,x_2}}{2S} \right) - \rho(y,x_2) \left( 1-\frac{\delta_{x_1,x_2}}{2S}\right)  \right]  \\ & \quad +  \sum_{\underset{|y-x_2|=1}{y\in\Lambda}}  \left[ \rho(x_1,x_2)\left( 1 - \frac {\delta_{x_1,y}}{2S}  \right) - \rho(x_1,y)  \left(  1-\frac{\delta_{x_1,x_2}}{2S} \right) \right] \,. \label{diff:ineq:lem}
\end{align}
\end{lemma}

\begin{proof}
The Heisenberg Hamiltonian (\ref{hamb}) can be written as 
\begin{equation}
H_\Lambda =  S \sum_{\langle x,y\rangle\subset \Lambda}  \left( a^\dagger_x \sqrt{ 1 - \frac {n_y}{2S} } - a^\dagger_y \sqrt{ 1-\frac{n_x}{2S}} \right) \left(  a_x \sqrt{ 1 - \frac {n_y}{2S} } - a_y \sqrt{ 1-\frac{n_x}{2S}} \right)\,,
\end{equation}
where $n_x = a^\dagger_x a_x$ and the sum is over all bonds in the graph. Equivalently, 
\begin{equation}
H_\Lambda =  S\sum_{(x,y) }  \left( a^\dagger_x \sqrt{ 1 - \frac {n_y}{2S} } - a^\dagger_y \sqrt{ 1-\frac{n_x}{2S}} \right)   a_x \sqrt{ 1 - \frac {n_y}{2S} } \,,
\end{equation}
where the sum is now over all {\em ordered} nearest neighbor pairs in $\Lambda$.

Let $\Psi$ be an eigenfunction of $H_\Lambda$ with eigenvalue $E$. Then
\begin{equation}
E \rho(x_1,x_2)  = E \langle \Psi | a^\dagger_{x_1} a^\dagger_{x_2} a_{x_2} a_{x_1} | \Psi \rangle = \langle \Psi | H_\Lambda a^\dagger_{x_1} a^\dagger_{x_2} a_{x_2} a_{x_1}  | \Psi \rangle \,.
\end{equation}
We compute
\begin{equation}
 a_x \sqrt{ 1 - \frac {n_y}{2S} }  \, a^\dagger_{x_1}  a^\dagger_{x_2} a_{x_2} a_{x_1} =  \left( a^\dagger_{x_1} a^\dagger_{x_2} a_{x_2} a_{x_1} + \delta_{x,x_1} n_{x_2} + \delta_{x,x_2} n_{x_1} \right)  a_x \sqrt{ 1 - \frac {n_y}{2S} }  
\end{equation}
 and thus
\begin{align}\nonumber 
E \rho(x_1,x_2)  & =   S \sum_{( x,y )  }  \biggl\langle \Psi \biggl|  \left( a^\dagger_x \sqrt{ 1 - \frac {n_y}{2S} } - a^\dagger_y \sqrt{ 1-\frac{n_x}{2S}} \right) 
\\ & \qquad\qquad \times  \left( 
a^\dagger_{x_1} a^\dagger_{x_2} a_{x_2} a_{x_1}  + \delta_{x,x_1} n_{x_2} + \delta_{x,x_2} n_{x_1} \right) a_x \sqrt{ 1 - \frac {n_y}{2S} }  \,\biggl| \Psi \biggl\rangle \,. \label{lfa}
\end{align}
The contribution of the first term $a^\dagger_{x_1} a^\dagger_{x_2} a_{x_2} a_{x_1}$ in the middle parenthesis is non-negative after summing over all pairs $(x,y)$, and can hence be dropped for a lower bound. For the remaining two terms, we write the last factor in (\ref{lfa})  as 
\begin{equation}
 a_x \sqrt{ 1 - \frac {n_y}{2S} }   = \frac 12 \left(  a_x \sqrt{ 1 - \frac {n_y}{2S} }  - a_y \sqrt{ 1 - \frac {n_x}{2S} } \right) +  \frac 12 \left(  a_x \sqrt{ 1 - \frac {n_y}{2S} }  + a_y \sqrt{ 1 - \frac {n_x}{2S} } \right) 
\end{equation}
and observe that the contribution of the first term yields again a non-negative expression. Hence  we get the lower bound 
\begin{align}\nonumber
 E \rho(x_1,x_2)    & \geq \frac S 2 \sum_{( x,y ) } \biggl\langle \Psi \biggl|  \left( a^\dagger_x \sqrt{ 1 - \frac {n_y}{2S} } - a^\dagger_y \sqrt{ 1-\frac{n_x}{2S}} \right)  \\ & \qquad \qquad \times  \left( \delta_{x,x_1} n_{x_2} + \delta_{x,x_2} n_{x_1} \right) \left(  a_x \sqrt{ 1 - \frac {n_y}{2S} } + a_y \sqrt{ 1 - \frac {n_x}{2S} } \right) \biggl| \Psi \biggl\rangle \,.
\end{align}
Since the right side is real, we only have to consider the hermitian part of the operator involved. This gives
\begin{align}\nonumber
& \frac {2 E}{S} \rho(x_1,x_2)   \\ \nonumber & \geq  \sum_{( x,y ) } \left\langle \Psi \left| a^\dagger_x \sqrt{ 1 - \frac {n_y}{2S} }   \left( \delta_{x,x_1} n_{x_2} + \delta_{x,x_2} n_{x_1} \right)  a_x \sqrt{ 1 - \frac {n_y}{2S} } \, \right| \Psi \right\rangle \\ \nonumber & \quad -  \sum_{( x,y ) } \left\langle \Psi \left|   a^\dagger_y \sqrt{ 1-\frac{n_x}{2S}}   \left( \delta_{x,x_1} n_{x_2} + \delta_{x,x_2} n_{x_1} \right) a_y \sqrt{ 1 - \frac {n_x}{2S} }  \, \right| \Psi \right\rangle \\ \nonumber & = \sum_{y: |y-x_1|=1}  \left\langle \Psi \left| \left[   a^\dagger_{x_1} a^{\dagger}_{x_2} a_{x_2} a_{x_1}   \left( 1 - \frac {n_y}{2S} \right) - a^\dagger_y a^\dagger_{x_2} a_{x_2} a_y \left( 1-\frac{n_{x_1}}{2S} \right) \right] \right| \Psi \right\rangle \\ \nonumber & \quad +  \sum_{y:|y-x_2|=1}  \left\langle \Psi \left| \left[ a^\dagger_{x_1} a^\dagger_{x_2} a_{x_2} a_{x_1}\left( 1 - \frac {n_y}{2S}  \right) - a^\dagger_{x_1}  a^\dagger_y a_y a_{x_1} \left(  1-\frac{n_{x_2}}{2S} \right) \right] \right| \Psi \right\rangle\\ \nonumber  & = \sum_{y: |y-x_1|=1} \left[ \rho(x_1, x_2) \left( 1 - \frac {\delta_{y,x_2}}{2S} \right) - \rho(y,x_2) \left( 1-\frac{\delta_{x_1,x_2}}{2S}\right)  \right]  \\ & \quad +  \sum_{y:|y-x_2|=1}  \left[ \rho(x_1,x_2)\left( 1 - \frac {\delta_{x_1,y}}{2S}  \right) - \rho(x_1,y)  \left(  1-\frac{\delta_{x_1,x_2}}{2S} \right) \right]\,. \label{diff:ineq}
\end{align}
\end{proof}

Instead of looking at $\rho(x_1,x_2)$, it will be convenient below to define $\sigma(x_1,x_2)$ by
\begin{equation}
\rho(x_1,x_2) = \sigma(x_1,x_2) \left( 1 - \frac{\delta_{x_1,x_2}}{2S} \right) \,.
\end{equation}
For $S\geq 1$ this defines $\sigma$ in terms of $\rho$ for every pair of points; for $S=1/2$ we  take $\sigma(x,x)=0$, i.e., $\sigma = \rho$. By plugging this ansatz into (\ref{diff:ineq:lem}) we obtain
\begin{align}\nonumber
 \frac {2 E}S \sigma(x_1,x_2) & \geq    \sum_{y: |y-x_1|=1} \left( \sigma(x_1, x_2)  -\sigma(y,x_2) \right)\left( 1 - \frac {\delta_{y,x_2}}{2S} \right) \\ & \quad +
\sum_{y:|y-x_2|=1}  \left( \sigma(x_1,x_2) - \sigma(x_1,y)\right) \left( 1 - \frac {\delta_{x_1,y}}{2S}\right) \,.\label{diff:sigma}
\end{align}
In particular,  $\sigma$ satisfies the inequality 
\begin{equation}\label{diffineq}
\left( \left(-\Delta_{x_1} - \Delta_{x_2}\right) \sigma \right)(x_1,x_2) \leq \frac {2 E}S \sigma(x_1,x_2) +  \frac 1 S \sigma(x_1,x_2) \chi_{|x_1-x_2|=1}\,,
\end{equation}
with $\Delta$ denoting the Laplacian on $\Lambda$.

\subsection{Step 2} 
Consider now a cubic lattice restricted to $\ell^3$ sites, $\Lambda_\ell := [0,\ell)^3 \cap \Z^3$. The inequality (\ref{diffineq}) holds for $(x_1,x_2)\in \Lambda_\ell\times \Lambda_\ell$. It can be extended to all of $\Z^6$ via reflection: For $z \in \{0,1,\dots, \ell-1\}$ and $m\in \Z$, let 
\begin{equation}\label{def:refl}
z_m = m \ell + \frac 12 (\ell -1)  + (-1)^m \left( z - \frac 12 (\ell -1) \right) \in \{m \ell, m\ell+1,\dots,  (m+1)\ell - 1\}
\end{equation}
denote the image of $z$ obtained after reflecting $m$ times at the boundary of the interval. One readily checks that
\begin{equation}\label{prop:refl}
z - w_m = (-1)^m \left( z_{(-1)^{m+1} m} - w\right)\,,
\end{equation}
which will be useful below.
We extend this to $z \in \Lambda_\ell \times \Lambda_\ell$ componentwise, and introduce the corresponding $z_m$ for $m\in \Z^6$. For any function $f$ on $\Lambda_\ell\times\Lambda_\ell$, we define a corresponding function $f^{R}$ on $\Z^6$ by 
\begin{equation}
f^R(z_m) = f(z)
\end{equation}
for all $m\in\Z^6$ and $z\in \Lambda_\ell\times\Lambda_\ell$. With $\chi$ denoting the characteristic function of the subset of $\Lambda_\ell\times\Lambda_\ell$ with $|x_1-x_2|=1$, we obtain from (\ref{diffineq}) the bound
\begin{equation}\label{iop}
\left( -\Delta_{\Z^6}  \sigma^R \right)(z) \leq \frac {2 E}S \sigma^R(z) +  \frac 1 S \sigma^R(z) \chi^R(z)
\end{equation}
for {\em all} $z=(x_1,x_2) \in \Z^6$, and with $\Delta_{\Z^6}$ now the usual Laplacian on the full space $\Z^6$. 

We bound the $\sigma^R$ in the last term on the right side of (\ref{iop}) simply by $\|\sigma^R\|_\infty = \|\sigma\|_\infty$.
For $E<6S$, we can write the resulting inequality  equivalently as 
\begin{equation}
\sigma^R(z) \leq (1 - E/(6S))^{-1} \left( \langle \sigma^R \rangle(z) + \frac 1 {12 S} \|\sigma\|_\infty \chi^R(z) \right)\,,
\end{equation}
where $\langle\,\rangle$ means averaging over nearest neighbors in $\Z^6$. If we iterate this bound $n$ times, we further obtain
\begin{equation}\label{insert}
\sigma^R(z) \leq (1-E/(6S))^{-n}  \left( \sum_{w\in \Z^6} P_n(z,w) \sigma^R(w)  + \frac 1{12 S} \| \sigma\|_\infty  \sum_{w\in\Z^6} Q_n(z,w) \chi^R(w) \right)\,,
\end{equation}
where $P_n(z,w)$ denotes the probability that a simple symmetric random walk on $\Z^6$ starting at $z$ ends up at $w$ in $n$ steps, and $Q_n = \sum_{j=0}^{n-1} P_j$. 

In the next step, we shall derive a simple upper bound on $P_n$ which will allow us to bound the first term on the right side of (\ref{insert}) in terms of the $1$-norm of $\sigma$. Moreover, we shall carefully evaluate the last term in (\ref{insert}) in order to show that it is strictly less than $\|\sigma\|_\infty$. It can thus be combined with the term on the left side of (\ref{insert}) to obtain the desired bound on the $\infty$-norm in terms of the $1$-norm. 

\subsection{Step 3} We shall first give a bound on the last term in (\ref{insert}). 
For any $F\geq 0$, we can bound $Q_n$ as  
\begin{align}\nonumber
Q_n(z,w) & \leq (1+ F/6)^{n-1} \sum_{j=0}^{n-1} (1+F/6)^{-j} P_j(z,w) \\ \nonumber 
&\leq (1+ F/6)^{n-1} \sum_{j=0}^{\infty} (1+ F/6)^{-j} P_j(z,w) \\ & = 12  (1 +F/6 )^{n}  \, (-\Delta_{\Z^6}+ 2F)^{-1}(z,w)\,. \label{ins1}
\end{align}
We are thus left with the task of deriving an upper bound on the quantity 
\begin{equation}\label{f1}
\sum_{w\in\Z^6} (-\Delta_{\Z^6}+ 2F)^{-1}(z,w) \chi^R(w) = \sum_{m\in\Z^6} \sum_{w\in \Lambda_\ell \times \Lambda_\ell }  (-\Delta_{\Z^6}+ 2F)^{-1}(z,w_m) \chi(w)\,.
\end{equation}
Using detailed properties of the resolvent of the Laplacian on $\Z^6$, we can obtain the following bound. Its proof will be given in the appendix.

\begin{lemma}\label{lem:new}
Let  
\begin{equation}\label{def:c4}
C_4:= \frac{\sqrt{3}-1}{192 \pi^3} \Gamma^2(\tfrac 1{24}) \Gamma^2(\tfrac {11}{24})  \approx 0.2527 
\end{equation}
and assume that  $z \in\Lambda_\ell\times\Lambda_\ell$ is a distance $d$ away from the complement of $\Lambda_\ell \times \Lambda_\ell$. Then  
\begin{equation}\label{neweq}
\sum_{w\in\Z^6} (-\Delta_{\Z^6}+ 2F)^{-1}(z,w) \chi^R(w)\leq -\frac 12  +  (3 + F/2)  \left[ C_4 + \frac 2 {\pi d }    \left(  \frac { 2 (1+\sqrt{F} )^{\ell/3}}{1- (1+\sqrt{F})^{-\ell/3} } \right)^3 \right] \,.
\end{equation}
\end{lemma}

The last term on the right side of (\ref{neweq}) is due to the finite size of $\Lambda_\ell$. It would be absent in infinite volume, in which case we could set $F=0$. 
It will be very important to note that 
\begin{equation}
3 C_4 - \frac 12 \approx 0.2582 < \frac 12 \leq S\,.
\end{equation}
It implies that, for our choice of $F\sim \ell^{-2}$ below, the expression on the right side of (\ref{neweq}) is strictly less than $S$ for large enough $d$. 

It remains to derive a bound on $P_n(z,w)$. The central limit theorem implies that, for large $n$, $P_n(z,w)$ behaves like $(3/(\pi n))^{3} e^{-3 \|z-w\|_2^2/n}$. In fact, we have the following explicit bound.

\begin{lemma}\label{lem:pn}
Let $b_0 \approx 1.942 $ denote the unique solution of 
\begin{equation}
\frac{6 b^2}{(\sinh b)^2} = b\,.
\end{equation}
Then 
\begin{equation}
P_n(z,w) \leq  \left( \frac{3\pi}{n} \right)^3 \exp\left(- b_0 \frac { \|z-w\|_2^2 }{2  n}    \right) \,.
\end{equation}
\end{lemma}

The proof of Lemma~\ref{lem:pn} is straightforward and we shall give it in the appendix.
From Lemma~\ref{lem:pn} we have the bound
\begin{align}\nonumber
\sum_{w\in\Z^6} P_n(z,w) \sigma^R(w) & \leq \|\sigma\|_1  \sum_{k\in \Z^6} \sup_{w\in \Lambda_\ell\times\Lambda_\ell} P_n(z,w-k\ell) \\ \nonumber
 & \leq  \|\sigma\|_1 \left( \frac{3\pi}{n} \right)^3 \sum_{k\in \Z^6} \sup_{z,w\in \Lambda_\ell\times\Lambda_\ell}  \exp\left(- b_0 \frac { \|z-w + k\ell\|_2^2 }{2  n}    \right) \\ & = \|\sigma\|_1 \left( \frac{3\pi}{n} \right)^3 \left( \sum_{k\in \Z} \sup_{z,w\in \{0,1,\dots,\ell-1\}}  \exp\left(- b_0 \frac { |z-w+k\ell|^2 }{2  n}\right) \right)^6\,,
\end{align}
where $\|\sigma\|_1 = \sum_{w\in \Lambda_\ell\times\Lambda_\ell} \sigma(w)$. 
We can bound the last exponential by $1$ for $|k|\leq 1$, and by $\exp(-b_0 \ell^2 (|k|-1)^2 /(2n))$ for $|k|\geq 2$. This gives
\begin{align}\nonumber
\sum_{w\in\Z^6} P_n(z,w) \sigma^R(w) & \leq  \|\sigma\|_1 \left( \frac{3\pi}{n} \right)^3 \left( 3 + 2 \sum_{m\geq 1}   \exp\left(- b_0 \frac { m^2 \ell^2 }{2  n}\right) \right)^6 \\  \nonumber &  \leq  \|\sigma\|_1 \left( \frac{3\pi}{n} \right)^3 \left( 3 + 2 \int_0^\infty  \exp\left(- b_0 \frac { m^2 \ell^2 }{2  n}\right) dm \right)^6 \\ & =  \|\sigma\|_1 \left( \frac{3\pi}{n} \right)^3 \left( 3 + \sqrt{\frac{2\pi n}{b_0\ell^2}} \right)^6 \label{ins2}\,.
\end{align}

If we insert the bounds obtained in (\ref{ins1}), (\ref{neweq}) and (\ref{ins2}) into (\ref{insert}), we  obtain
\begin{align}\nonumber
\sigma(z) &\leq (1-E/(6S))^{-n}  \|\sigma\|_1 \left( \frac{3\pi}{n} \right)^3 \left( 3 + \sqrt{\frac{2\pi n}{b_0\ell^2}} \right)^6  \\ & \quad  + \left(\frac{1+F/6}{1-E/(6S)}\right)^{n}  \frac 1{ S} \| \sigma\|_\infty 
 \left( 3 C_4 -\frac 12 + \frac {C_4} 2 F +  \frac{6 + F}{\pi d }    \left[  \frac { 2 (1+\sqrt{F})^{\ell/3}}{1- (1+\sqrt{F})^{-\ell/3} } \right]^3\right) \label{54}
\end{align}
for all $z=(x_1,x_2) \in \Lambda_\ell\times\Lambda_\ell$  a distance $d$ away from its complement. The bound holds for all $n\geq 1$ and all $F>0$. 

We shall simply choose $F=\ell^{-2}$ and, recalling that $E/S\ge C\ell^{-2}$, we fix  
\begin{equation}
n = \lfloor \epsilon  S E^{-1}  \rfloor
\end{equation}
with $\epsilon$ small enough such that
\begin{equation}
\left(\frac{1+F/6}{1-E/(6S)}\right)^{n} \leq \frac{1-\delta}{6 C_4 -1 + C_4 \ell^{-2}}
\end{equation}
for some $\delta>0$ and all small  enough $E/S$. Since $6C_4-1 \approx 0.516 < 1$, this condition can be satisfied for small enough (but strictly positive) $\delta$.
The resulting bound is then
\begin{equation}\label{pp}
\sigma(z) \leq C S^{-3} E^{3}\|\sigma\|_1 + \frac 1 {2S} \left(1-\delta + C d^{-1} \right) \|\sigma\|_\infty \,.
\end{equation}

For $S$ large enough, the coefficient in front of the last term in (\ref{pp}) is smaller than $1$ for all $d\geq 1$, hence we obtain the desired result directly from (\ref{pp}) in this case, taking the supremum over $z$ on the left. For smaller $S$, we need an additional argument, which is provided in the next and final step.

\subsection{Step 4} The following lemma implies that $\sigma$ is very flat near its maximum. In particular, the maximal value of $\sigma$ in the smaller cube a distance $d$ away from the boundary of $\Lambda_\ell\times\Lambda_\ell$ is close to its global maximum. We shall deduce this property from the differential inequality (\ref{diff:sigma}). 

\begin{lemma}\label{lem:sig}
 Assume that $\sigma$ satisfies (\ref{diff:sigma}), and let 
$z_0 \in \Lambda_\ell\times \Lambda_\ell$ be such that $\sigma(z_0)=\|\sigma\|_\infty$. Then, for $S\geq 1$, 
\begin{equation}\label{s1b}
\min_{z: \|z-z_0\|_1=n} \sigma(z) \geq \|\sigma\|_\infty \left(1 - \frac {2 E}{11 S} \left(\frac{12}{1-\frac 1{2S}}\right)^{n} \right)
\end{equation}
for any $n\geq 1$. For $S=1/2$ we have the bound 
\begin{equation}\label{s12b}
\min_{z: d(z,z_0)=n} \sigma(z) \geq \|\sigma\|_\infty \left(1 -  \frac {4 E}{11} (12)^{n} \right)
\end{equation}
instead, where $d(z,w)$ denotes the distance on the graph $\Lambda_\ell\times\Lambda_\ell \setminus \{ (x,x) : x\in \Lambda_\ell\}$.
\end{lemma}

\begin{proof}
Let us first consider the case $S\geq 1$. 
Let 
\begin{equation}
\nu_n = \|\sigma\|_\infty^{-1} \min_{z: \|z-z_0\|_1=n} \sigma(z) \,,
\end{equation}
and choose $z_n$ with $\|z_n-z_0\|_1=n$. Let us define the degree of the vertex $z\in \Lambda_\ell\times \Lambda_\ell$ as 
\begin{equation}
d_{z} = \sum_{w: |w-z|=1} \left( 1- \frac{ \delta_{w_1,w_2}}{2S}\right)\,.
\end{equation}
The inequality (\ref{diff:sigma}) can be written as 
\begin{equation}
\frac {2 E}S \sigma(z) \geq  d_z \sigma(z) - \sum_{w: |w-z|=1} \sigma(w) \left( 1- \frac{ \delta_{w_1,w_2}}{2S}\right)\,.
\end{equation}
Hence we have, for $z=z_n$, 
\begin{equation}
\frac {2 E}S \|\sigma\|_\infty \geq \frac {2 E}S \sigma(z_n) \geq d_{z_n} \sigma(z_n) - (d_{z_n}-\lambda) \|\sigma\|_\infty - \lambda \sigma(z_{n+1})\,,
\end{equation}
where $z_{n+1}$ is a neighbor of $z_n$ such that $\|z_{n+1}-z_0\|_1=n+1$, and $\lambda$ is either $1$ or $(1-1/(2S))$. Equivalently, 
\begin{equation}
\sigma(z_{n+1}) \geq \frac 1 \lambda d_{z_n}  \sigma(z_n) - \frac 1 \lambda (d_{z_n}+ 2 E S^{-1}-\lambda) \|\sigma\|_\infty\,.
\end{equation}
Note that $d_z \leq 12$ for any $z$. The right side above is decreasing in $d_{z_n}$ and increasing in $\lambda$, hence we can replace $d_{z_n}$ by $12$ and $\lambda$ by $(1-1/(2S))$ for a lower bound. Moreover, we pick $z_n$ and $z_{n+1}$ in such a way that $\sigma(z_{n+1}) = \|\sigma\|_\infty \nu_{n+1}$. 
We thus conclude that
\begin{equation}
\nu_{n+1}\geq 1  - \frac { 12 (1 - \nu_n) + 2 E S^{-1}} { 1 - \frac 1{2S} } \,.
\end{equation}
By induction, one easily sees that this implies that 
\begin{equation}
\nu_n \geq 1 - \frac {2 E}S \frac { \left( \frac{12}{1-\frac 1{2S}}\right)^n -1}{11 + \frac 1{2S}} \geq 1 -   \frac {2E}{11S}  \left( \frac {12}{1-\frac 1{2S}} \right)^{n}\,.
\end{equation}
This proves the bound (\ref{s1b}) in the case $S\geq 1$. The proof of (\ref{s12b}) works analogously.
\end{proof}

Lemma~\ref{lem:sig} implies that
\begin{equation}
\sup \{ \sigma(z) \, : \, \dist(z,(\Lambda_\ell\times\Lambda_\ell)^c) \geq d\} \geq \|\sigma\|_\infty \left( 1 - \frac {2E}{11S} \left( \frac {12}{1-\frac 1{2S}} \right)^{d-1} \right)
\end{equation}
for $S\geq 1$. Similarly, we can bound for $S=1/2$
\begin{equation}
\sup \{ \sigma(z) \, : \, \dist(z,(\Lambda_\ell\times\Lambda_\ell)^c) \geq d\} \geq \|\sigma\|_\infty \left( 1 - \frac {4E}{11} (12)^{d+1} \right)\,,
\end{equation}
noting that because of the hard-core constraint $z_1\neq z_2$ it may take up to two more steps to go from a point $w$ to a point $z$. In both cases,
\begin{equation}
\sup \{ \sigma(z) \, : \, \dist(z,(\Lambda_\ell\times\Lambda_\ell)^c) \geq d\} \geq \|\sigma\|_\infty \left( 1 -  E S^{-1} C^d  \right)
\end{equation}
for a constant $C>1$. We plug this into (\ref{pp}), taking the maximum over all $z$ a distance $d$ away from the boundary on the left side. This gives
\begin{equation}\label{pp2}
\|\sigma\|_\infty \left( 1 -  \frac {E C^d}S   - \frac 1{ 2 S} \left(1-\delta + C d^{-1} \right) \right) \leq CS^{-3} E^{3} \|\sigma\|_1  \,,
\end{equation}
and this bound now holds for all $d$. We choose $d$ large enough such that $1- \frac 1{2 S}(1-\delta +C d^{-1}) \geq \delta/2$, and thus obtain, for small enough $E/S$, 
\begin{equation}
\|\sigma\|_\infty \leq CS^{-3}E^{3} \|\sigma\|_1  \,.
\end{equation}
Since $\rho(z)\leq \sigma(z) \leq \rho(z) (1-1/(2S))^{-1}$ for $S\geq 1$, and $\rho(z)=\sigma(z)$ for $S=1/2$, this implies (\ref{thm:eq:rho}). 

\bigskip

\noindent
{\bf Acknowledgments.} The research leading to these results has received funding from the European Research Council under the European Union's Seventh Framework Programme ERC Starting Grant CoMBoS (grant agreement 
n$^o$ 239694).

\appendix
\section{Proofs of Auxiliary Lemmas}

\begin{proof}[Proof of Lemma~\ref{lem:new}]

From the property (\ref{prop:refl}) and translation-invariance and parity of the Laplacian, the expression in (\ref{f1})  equals
\begin{equation}\label{f2}
(\ref{f1}) = \sum_{m\in\Z^6} \sum_{w\in \Lambda_\ell \times \Lambda_\ell }  (-\Delta_{\Z^6}+ 2F)^{-1}(z_m,w) \chi(w)
\end{equation}
for $z\in \Lambda_\ell\times\Lambda_\ell$. Because of $\chi$, the sum is restricted to $w=(w_1,w_2)\in \Lambda_\ell\times\Lambda_\ell$ with $|w_1-w_2|=1$. Since the resolvent of the Laplacian has a positive kernel, we can drop the condition that $w_2\in \Lambda_\ell$ for an upper bound. This gives
\begin{equation}\label{f3}
\eqref{f2} \leq \sum_{m\in\Z^6} \sum_{e\in\Z^3 : |e|=1} \sum_{x \in \Lambda_\ell }  (-\Delta_{\Z^6}+ 2F)^{-1}(z_m,(x,x+e)) \,.
\end{equation}
The resolvent of the Laplacian can be conveniently written in terms of its Fourier transform as 
\begin{equation}
 (-\Delta_{\Z^6}+ 2F)^{-1}(z,w) = \frac 1{(2\pi)^{6}} \int_{[-\pi,\pi]^6} \frac {e^{i p_1 \cdot (x_1-y_1)+ i p_2 \cdot (x_2-y_2)}} { \epsilon(p_1) + \epsilon(p_2) +2 F }  dp_1\, dp_2\,,
\end{equation}
where $z=(x_1,x_2)$, $w=(y_1,y_2)$ and $\epsilon(p) = 6 -\sum_{e\in\Z^3: |e|=1} e^{i p \cdot e}$ denotes the dispersion relation of the Laplacian on $\Z^3$. Hence  
\begin{align}\nonumber 
\eqref{f3} & = \sum_{m_1\in\Z^3} \sum_{m_2\in \Z^3} \sum_{e\in\Z^3 : |e|=1} \sum_{x\in \Lambda_\ell}  \frac 1{(2\pi)^{6}} \int_{[-\pi,\pi]^6} \frac {e^{i p_1 \cdot (x_{1,m_1}- x ) + i p_2 \cdot (x_{2,m_2} - x - e)}} { \epsilon(p_1) + \epsilon(p_2)+2F }  dp_1\, dp_2 \\   &  = \sum_{m_1\in\Z^3} \sum_{m_2\in \Z^3}  \sum_{x\in \Lambda_\ell}  \frac 1{(2\pi)^{6}} \int_{[-\pi,\pi]^6} \frac {e^{i p_1\cdot (x_{1,m_1} - x )+i p_2 \cdot (x_{2,m_2} - x)}} { \epsilon(p_1) + \epsilon(p_2) +2F } \left(6 - \epsilon(p_2) \right) dp_1\, dp_2  \,. \label{f4}
\end{align}
With the aid of the identity (\ref{prop:refl}), we can rewrite the last expression as 
\begin{align}\nonumber 
\eqref{f4} & = \sum_{m_1\in\Z^3} \sum_{m_2\in \Z^3}  \sum_{x\in \Lambda_\ell} \frac 1{ (2\pi)^{6}} \int_{[-\pi,\pi]^6} \frac {e^{i p_1\cdot (x_{1} - x_{m_1} )+i p_2\cdot (x_{2,m_2} - x_{m_1})} } { \epsilon(p_1) + \epsilon(p_2) +2 F} \left(6 - \epsilon(p_2) \right)dp_1\, dp_2\\ &  =  \frac 12\sum_{m_2\in \Z^3}  \frac 1  {(2\pi)^{3}} \int_{[-\pi,\pi]^3} \frac {6 - \epsilon(p)} { \epsilon(p) + F} e^{i p\cdot (x_{1} - x_{2,m_2})} dp \,. \label{f5a}
\end{align}
For $(x_1,x_2)\in \Lambda_\ell\times \Lambda_\ell$, this further equals
\begin{equation}\label{f5}
\eqref{f5a} = (3+F/2) \sum_{m\in\Z^3} \left( -\Delta_{\Z^3} + F \right)^{-1}(x_1,x_{2,m}) - \frac 12 \delta_{x_1,x_2} \,.
\end{equation}
At this point, we need some properties of the resolvent of the Laplacian on $\Z^3$, which we collect in the following lemma. Its proof will be given at the end of the proof of Lemma~\ref{lem:new}

\begin{lemma}\label{lem:gre}
For $F\geq 0$, the function $\Z^3 \ni x \mapsto \left( -\Delta_{\Z^3} + F \right)^{-1} (0,x)$ is positive and decreasing in the components $x^j$ for $x^j$ positive, and increasing otherwise. We have the bounds  
\begin{equation}\label{value00}
\left( -\Delta_{\Z^3} + F \right)^{-1} (0,0)\leq \left( -\Delta_{\Z^3} \right)^{-1}(0,0) = \frac{\sqrt{3}-1}{192 \pi^3} \Gamma^2(\tfrac 1{24}) \Gamma^2(\tfrac {11}{24})  \approx 0.2527 
\end{equation}
and 
\begin{equation}\label{value01}
\left( -\Delta_{\Z^3} + F \right)^{-1} (0,x)\leq \left( -\Delta_{\Z^3} \right)^{-1}(0,e)  = \frac{\sqrt{3}-1}{192 \pi^3} \Gamma^2(\tfrac 1{24}) \Gamma^2(\tfrac {11}{24})  - \frac 16 \approx 0.0861
\end{equation}
for $x \neq 0$ and $|e|=1$. Moreover, for $x\neq  0$ and $\|x\|_\infty = \max_{1\leq j \leq 3} |x^j|$, 
\begin{equation}\label{exp:decay}
\left( -\Delta_{\Z^3} + F \right)^{-1} (0,x)\leq  \frac 2 {\pi \|x\|_\infty} \left( 1  + \sqrt{F} \right)^{-\|x\|_\infty} \,.
\end{equation}
\end{lemma}

With $C_4$ defined in (\ref{def:c4}), Lemma~\ref{lem:gre} implies that 
\begin{equation}
 (3+F/2) \left( -\Delta_{\Z^3} + F \right)^{-1}(x_1,x_2) - \frac 12 \delta_{x_1,x_2}  \leq C_4 \left(3 +F/2\right) -\frac 12 \,.
\end{equation}
Moreover, if $x_1\in\Lambda_\ell$ is at least a distance $d$ from the complement of $\Lambda_\ell$, we can use (\ref{exp:decay}) to bound the contribution of $m\neq 0$ to the sum in (\ref{f5}). Since $\|x_1 -x_{2,m}\|_\infty \geq d + (\|m\|_\infty-1) \ell$ in this case, this gives 
\begin{align}\nonumber 
\eqref{f5} & \leq  -\frac 12 +  (3 + F/2) \left[ C_4 +  \sum_{m\neq 0} \frac 2 {\pi \|x_1 -x_{2,m}\|_\infty} \left( 1  + \sqrt{F} \right)^{-\|x_1-x_{2,m}\|_\infty} \right] \\ & \leq -\frac 12  +  (3 + F/2) \left[ C_4 +  \frac 2 {\pi d } \sum_{m\neq 0}  \left( 1  + \sqrt{F} \right)^{- (\|m\|_\infty-1)\ell} \right] \,. \label{f6}
\end{align}
Using $\|m\|_\infty \geq \|m\|_1/3$ in the last sum, we obtain the desired bound (\ref{neweq}), with $d$ the distance of $x_1$ to the complement of $\Lambda_\ell$. This distance is greater or equal to the distance of $z=(x_1,x_2)$ to the complement of $\Lambda_\ell\times\Lambda_\ell$, hence the proof is complete.
\end{proof}

\begin{proof}[Proof of Lemma~\ref{lem:gre}]
The resolvent of the Laplacian on $\Z^3$ can be expressed via the heat kernel  as (see, e.g., \cite{CY})
\begin{equation}
\left( -\Delta_{\Z^3} + F\right)^{-1}(x,y) = \int_{0}^\infty  e^{-6t} I_{|x^1-y^1|}(2t) I_{|x^2-y^2|}(2t) I_{|x^3-y^3|}(2t) e^{-F t}  dt 
\end{equation}
for $F\geq 0$, with $I_n$ denoting the modified Bessel functions, which are positive and increasing on the positive real axis. (For a definition, see \cite{AS} or Eq.~(\ref{def:is}) below.) The monotonicity property of the resolvent then follows directly from the fact that $I_n(t) \geq I_{n+1}(t)$ for all $t\in\R$ and $ n \in \N$.  To see this last property, note that the recursion relations \cite[9.6.26]{AS} imply that $R_n(t) = I_{n}(t)-I_{n+1}(t)$ satisfies
\begin{equation}
R_n(t)= \frac n t I_n(t) + \frac{n+1}{t} I_{n+1}(t) - R_n'(t) \,.
\end{equation}
This further implies that $R_n'$ is positive whenever $R_n$ is zero. Since $R_n$ is positive for small argument, as can be seen from the asymptotic expansion \cite[9.6.10]{AS}, for instance, this is impossible. Hence $R_n$ is positive. 

The values of the integrals corresponding to (\ref{value00}) and (\ref{value01}) can be found in \cite[6.612(6)]{GR}. Finally, to obtain the bound (\ref{exp:decay}), we start with the integral representation \cite[9.6.18]{AS}
\begin{equation}\label{def:is}
I_n(t) = \frac{ (t/2)^n}{\sqrt\pi \Gamma(n+1/2)} \int_{-1}^1 (1-s^2)^{n-1/2} e^{-st}ds\,.
\end{equation}
It implies that 
\begin{align}\nonumber
I_0(t) &= \frac{ 1}{\pi}  \int_{-1}^1 (1-s^2)^{-1/2} e^{-st}ds 
=  \frac{ e^t}{\pi}  \int_{0}^1 \frac{ e^{-st}}{\sqrt{s}} \frac {1+e^{-2(1-s)t}}{\sqrt{2-s}} ds \\ & 
\leq \frac{ 2e^t}{\pi}  \int_{0}^1 \frac{ e^{-st}}{\sqrt{s}} ds   \leq  \frac{ 2e^t}{\pi}  \int_{0}^\infty \frac{ e^{-st}}{\sqrt{s}}ds = \frac{ 2e^t}{\sqrt{\pi t}}\,.  \label{i0e} 
\end{align}
Hence, with $n=\|x-y\|_\infty$, we further have 
\begin{align}\nonumber 
\left( -\Delta_{\Z^3} + F\right)^{-1}(x,y) & \leq  \int_{0}^\infty  e^{-6t} I_{n}(2t) I_{0}(2t)^2 e^{-Ft}  dt \\ \nonumber
&  \leq  \frac 2 \pi \int_{0}^\infty \frac 1t  e^{-2t} I_{n}(2t)   e^{-Ft}  dt \\ & = \frac 2 {\pi n} \left( 1 + F/2 + \sqrt{F(1+F/4)} \right)^{-n} \leq \frac 2 {\pi n} \left( 1  + \sqrt{F} \right)^{-n} \,,
\end{align}
where we used \cite[6.623(3)]{GR} to compute the integral.
\end{proof}

\begin{proof}[Proof of Lemma~\ref{lem:pn}]
We start with the integral representation
\begin{equation}
P_n(z,w) = \frac 1{ (2\pi)^{6}} \int_{[-\pi,\pi]^6}  \left( \frac 16 \sum_{j=1}^6 \cos(q_j) \right)^n e^{i \sum_j q_j (z^j-w^j)} dq_1 \cdots dq_6\,.
\end{equation}
The integrand  is a Laurent polynomial in the $e^{iq_j}$, and hence the integral does not change if the $q_j$ are replaced by $q_j+i a_j$ for any $a_j\in \C$. We shall choose $a_j\in \R$, and bound
\begin{align}\nonumber
P_n(z,w) & = \frac 1{(2\pi)^{6}} \int_{[-\pi,\pi]^6}  \left( \frac 16 \sum_{j=1}^6 \cos(q_j+i a_j) \right)^n e^{i \sum_j (q_j+i a_j) (z^j-w^j)} dq_1 \cdots dq_6 \\ \nonumber  & \leq \frac 1 {(2\pi)^{6}} \int_{[-\pi,\pi]^6}  \left( \frac 16 \sum_{j=1}^6 \left| \cos(q_j+i a_j) \right| \right)^n e^{- \sum_j a_j (z^j-w^j)} dq_1 \cdots dq_6 \\ & = \pi^{-6} e^{- \sum_j a_j (z^j-w^j)} \int_{[-\pi/2,\pi/2]^6}  \left( \frac 16 \sum_{j=1}^6 \left| \cos(q_j+i a_j) \right| \right)^n dq_1 \cdots dq_6\,. \label{plug}
\end{align}
We have
\begin{equation}
|\cos(q+ia)|^2 = (\sinh a)^2 + (\cos q)^2 \leq \left( \frac {a \sinh b}{b} \right)^2 + 1 - (2q/\pi)^2
\end{equation}
for $|a|\leq b$ and  $|q|\leq \pi/2$. In particular,
\begin{align}\nonumber 
\frac 16 \sum_{j=1}^6 \left| \cos(q_j+i a_j) \right| & \leq \left( \frac 16 \sum_{j=1}^6 \left| \cos(q_j+i a_j) \right|^2 \right)^{1/2} \\ \nonumber
 & \leq \left(  1 - \frac 2{3\pi^2} \sum_{j=1}^6 q_j^2 + \frac{(\sinh b)^2}{6 b^2} \sum_{j=1}^6 a_j^2  \right)^{1/2} \\ & \leq \exp\left( - \frac 1{3\pi^2} \sum_{j=1}^6 q_j^2 + \frac{(\sinh b)^2}{12 b^2} \sum_{j=1}^6 a_j^2\right)\,.
\end{align}
Plugging this bound into (\ref{plug}), we obtain
\begin{align}\nonumber
P_n(z,w) & \leq \exp\left(- \sum_{j=1}^6 a_j (z^j-w^j) + n \frac{(\sinh b)^2}{12 b^2} \sum_{j=1}^6 a_j^2\right)  \left( \frac 1 \pi \int_{[-\pi/2,\pi/2]} e^{ - n q^2/(3\pi^2)} 
 dq \right)^6 \\ &\leq \left( \frac{3\pi}{n} \right)^3 \exp\left(- \sum_j a_j (z^j-w^j) + n \frac{(\sinh b)^2}{12 b^2} \sum_{j=1}^6 a_j^2\right) \,.
\end{align}

To minimize the right side, we choose 
\begin{equation}
a_j = \frac { (z^j-w^j)}{n} \frac{6 b^2}{(\sinh b)^2}\,.
\end{equation}
Keeping in mind that $|a_j|\leq b$ is required for all $j$, we see that if  
\begin{equation}
\frac{6 b^2}{(\sinh b)^2} \leq b
\end{equation}
then we obtain the bound 
\begin{equation}
P_n(z,w) \leq  \left( \frac{3\pi}{n} \right)^3 \exp\left(- \frac { \|z-w\|_2^2 }{ n} \frac{3 b^2}{(\sinh b)^2}   \right) 
\end{equation}
for all $z$ and $w$ with $\|z-w\|_\infty\leq n$. But $P_n(z,w)=0$ for $\|z-w\|_\infty > n$, hence this establishes the desired bound for all values of $z\in \Z^6$ and $w\in \Z^6$.
\end{proof}

\section{Quasi long-range order}\label{appB}

Here we prove \eqref{2.or}. Let us preliminarily observe that taking the expectation of \eqref{indi2} in $\langle\cdot\rangle_\beta$ we immediately get \eqref{2.or} with the factor $9/8$ replaced by 2. To improve the factor, let 
\begin{equation}f_n=\sup_{\|x-y\|_1=n}\langle S^2-\vec S_x\cdot\vec S_y\rangle_\beta\;.\end{equation}
Note that $f_1=e(S,\beta)$. Using \eqref{ineq 1}, for $n>1$,
\begin{equation} f_n\le 2(f_{n-j}+f_j)\;,\qquad 1\le j<n.\end{equation}
We pick $j=\lfloor n/2\rfloor$, so that in particular
\begin{equation} f_n\le \begin{cases} 4f_{n/2}\;,& {\rm if}\ n\ {\rm is}\ {\rm even}\\
2(f_{\frac{n-1}2}+f_{\frac{n+1}2})\;,& {\rm if}\ n\ {\rm is}\ {\rm odd}\;.\end{cases}\label{b.it}\end{equation}
We claim that this implies that $f_n\le \tfrac98 n^2 f_1$. To see this, observe that the solution $g_n$ to the iteration defined by \eqref{b.it} with  equality sign and initial datum $f_1=1$ is
\begin{equation} g_n=n^2+2^k m-m^2\;,\qquad {\rm if}\ n=2^k+m\quad {\rm with}\quad 0\le m< 2^k\;.\end{equation}
Hence $f_n\le g_n f_1$ for all $n$. Given $n=2^k+m$ with $0\le m<2^k$, one has $g_n/n^2=(3x+1)/(x+1)^2$ with $x=m 2^{-k}$. Maximizing over $x$ gives $g_n\le \tfrac98 n^2$,
from which \eqref{2.or} follows.

\end{document}